\documentclass[journal]{IEEEtran}

\usepackage{graphicx}
\usepackage{amsfonts}
\usepackage{amsmath}
\usepackage{amssymb}
\usepackage{color}
\usepackage{bm} 

\newtheorem{thm}{Theorem}
\newtheorem{cor}[thm]{Corollary}
\newtheorem{lem}[thm]{Lemma}
\newtheorem{proof}[thm]{proof}

\newtheorem{defn}[thm]{Definition}
\newtheorem{rem}[thm]{Remark}
\newtheorem{exam}[thm]{Example}

\begin{document}
\onecolumn 

\title{Locally Repairable Codes with Functional Repair and Multiple
Erasure Tolerance}

\author{Wentu Song ~ and ~ Chau Yuen
\thanks{W. Song and C. Yuen are with Singapore University of
Technology and Design, Singapore
       (e-mails: \{wentu\_song, yuenchau\}@sutd.edu.sg).}
}
\maketitle

\begin{abstract}
We consider the problem of designing $[n,k]$ linear codes for
distributed storage systems (DSS) that satisfy the
$(r,t)$-\emph{Local Repair Property}, where any $t' (\leq t)$
simultaneously failed nodes can be locally repaired, each with
locality $r$. The parameters $n,k,r,t$ are positive integers such
that $r<k<n$ and $t\leq n-k$. We consider the functional repair
model and the sequential approach for repairing multiple failed
nodes. By functional repair, we mean that the packet stored in
each newcomer is not necessarily an exact copy of the lost data
but a symbol that keep the $(r,t)$-local repair property. By the
sequential approach, we mean that the $t'$ newcomers are ordered
in a proper sequence such that each newcomer can be repaired from
the live nodes and the newcomers that are ordered before it. Such
codes, which we refer to as $(n,k,r,t)$-functional locally
repairable codes (FLRC), are the most general class of LRCs and
contain
several subclasses of LRCs reported in the literature.

In this paper, we aim to optimize the storage overhead
(equivalently, the code rate) of FLRCs. We derive a lower bound on
the code length $n$ given $t\in\{2,3\}$ and any possible $k,r$.
For $t=2$, our bound generalizes the rate bound proved in
\cite{Prakash-14}. For $t=3$, our bound improves the rate bound
proved in \cite{Tamo14}. We also give some constructions of exact
LRCs for $t\in\{2,3\}$ whose length $n$ achieves the bound of
$(n,k,r,t)$-FLRC, which proves the tightness of our bounds and
also implies that there is no gap between the optimal code length
of functional LRCs and exact LRCs for certain sets of parameters.
Moreover, our constructions are over the binary field, hence are
of interest in practice.
\end{abstract}

\begin{IEEEkeywords}
Distributed storage, locally repairable codes, exact repair,
functional repair.
\end{IEEEkeywords}

\IEEEpeerreviewmaketitle \setcounter{equation}{0}

\section{Introduction}

A distributed storage system (DSS) stores data through a large,
distributed network of storage nodes. To ensure reliability
against node failure, data is stored in redundancy form so that it
can be reconstructed from the system even if some of the storage
nodes fail. Moreover, to maintain the data reliability in the
presence of node failures, each failed node is replaced by a
\emph{newcomer} that stores a data packet computed from the data
packets stored in some available storage nodes. This process is
called \emph{node repair}.

There are two models of node repair, called \emph{exact repair}
and \emph{functional repair} respectively. By exact repair, each
newcomer stores an exact copy of the lost data packet. By
functional repair, each newcomer stores a packet that is not
necessarily an exact copy of the lost data, but a packet that
makes the system keep the same level of data reliability and the
possibility of node repair in the future. While exact repair is a
special case of functional repair and is more preferable in
practice for its simplicity, functional repair model has its
theoretical interest because potentially it allows us to construct
codes with improved code rate or minimum distance.

Modern distributed storage systems employ various coding
techniques, such as erasure codes, regenerating codes and locally
repairable codes, to improve system efficiency. Classical MDS
codes (such as Reed-Solomon codes) are optimal in storage
efficiency but are inefficient in node repair---the total amount
of data download needed to repair a single failed node equals to
the size of the whole file \cite{Dimakis10}. As improvements of
MDS codes, regenerating codes aim to optimize the repair bandwidth
\cite{Dimakis10} and locally repairable codes (LRC) aim to
minimize the repair locality, i.e. the number of disk accesses
required during a single node repair \cite{Papail122}. In this
work, we focus on the metric of repair locality.

Repair locality was initially studied as a metric for repair cost
independently by Gopalan et al. \cite{Gopalan12}, Oggier et al.
\cite{Oggier11}, and Papailiopoulos et al. \cite{Papail121}. The
$i$th coordinate of an $[n, k]_q$ linear code $\mathcal C~($also
called the $i$th code symbol of $\mathcal C)$ is said to have
locality $r$, if its value is computable from the values of a set
of at most $r$ other coordinates of $\mathcal C~($called a repair
set of $i)$. In the literature, an $[n, k]$ linear code is called
a locally repairable code (LRC) if all of its code symbols have
locality $r$ for some $r<k$. In a DSS coded by an LRC $\mathcal
C$, each storage node stores a code symbol of $\mathcal C$ and any
single failed node can be ``locally and exactly repaired" in the
sense that the newcomer can recover the lost data by contacting at
most $r$ live nodes, where $r$ is the symbol locality of $\mathcal
C$.

\subsection{Local Repair for Multiple Node Failures}
In real DSS, it is not uncommon that two or more storage nodes
fail simultaneously at one time, which motivates the researchers
to study LRCs that can locally repair more than one failed nodes.
Studies of LRCs for multiple node failures can be found in
\cite{Pamies13}$-$\cite{Wentu14} and references therein.

To repair $t~(t\geq 2)$ simultaneously failed nodes, $t$ newcomers
are added into the system, each downloads data from a set of at
most $r$ available nodes to create its storage content. The
authors in \cite{Prakash-14} distinguished two approaches of how
the $t$ newcomers contact the available nodes, called
\emph{parallel approach} and \emph{sequential approach}
respectively. By the parallel approach, each newcomer download
data from a set of live nodes. In contrast, by the sequential
approach, the $t$ newcomers can be properly ordered in a sequence
and each newcomer can download data from both the live nodes and
the newcomers ordered before it. Clearly, the parallel approach is
a special case of the sequential approach. Potentially, the
sequential approach allows us to design codes with improved code
rate or minimum distance than the parallel approach.

Given the parameters $n,k,r$ and $t$, where $n$ is the code length
and $k$ is the dimension, four subclasses of linear LRCs that can
exactly and locally repair up to $t$ failed nodes by the parallel
approach are reported in the literature: a) Codes with all-symbol
locality $(r,t+1)$, in which each code symbol is contained in a
local code of length at most $r+t$ and minimum distance at least
$t+1$ \cite{Prakash12}; b) Codes with all-symbol locality $r$ and
availability $t$, in which each code symbol has $t$ pairwise
disjoint repair sets with locality $r$ \cite{Wang14,Rawat14}; c)
Codes with $(r,t)$-locality, in which each subset of $t$ code
symbols can be cooperatively repaired from at most $r$ other code
symbols \cite{Rawat-14}; d) Codes with overall local repair
tolerance $t$, in which for any $E\subseteq[n]$ of size $t$ and
any $i\in E$, the $i$th code symbol has a repair set contained in
$[n]\backslash E$ and with locality $r$ \cite{Pamies13}.

For convenience, we refer to the above four subclasses of LRCs as
$(r,\delta)_a$ codes, $(r,\delta)_c$ codes, $(r,t)$-CLRC and
$(r,t)_o$ codes respectively, where $\delta=t+1$. Clearly, the
first three subclasses are all contained in the subclass of
$(r,t)_o$ codes. Moreover, $(r,t)_o$ codes can exactly and locally
repair up to $t$ failed nodes by the parallel approach. For
$(r,\delta)_a$ codes and $(r,t)$-CLRC, the code rate satisfies
(e.g., see \cite{Wentu14} and \cite{Rawat-14}):
\begin{align}\label{rate-bd-1}\frac{k}{n}\leq
\frac{r}{r+t}\end{align} and the minimum distance satisfies (see
\cite{Prakash12} and \cite{Rawat-14}):
\begin{align}\label{d-bd-1}d\leq
n-k+1-t\left(\left\lceil\frac{k}{r}\right\rceil-1\right).
\end{align}
For $(r,\delta)_c$ codes $($i.e., codes with all-symbol locality
$r$ and availability $t)$, it was proved in \cite{Tamo14} that the
code rate satisfies:
\begin{align}\label{rate-bd-2}\frac{k}{n}\leq
\frac{1}{\prod_{j=1}^{t}(1+\frac{1}{jr})}\end{align} and the
minimum distance satisfies:
\begin{align}\label{d-bd-2}d\leq
n-\sum_{i=0}^{t}\left\lfloor\frac{k-1}{r^i}\right\rfloor.
\end{align}
For $t=2$, the bound \eqref{d-bd-2} is shown to be achievable for
some special case of parameters \cite{Tamo14}. However, for the
general case, it is not known whether the bounds \eqref{rate-bd-2}
and \eqref{d-bd-2} are achievable. Recent work by Wang et al.
\cite{Wang15} shows that for any positive integers $r$ and $t$,
there exist $(r,\delta)_c$ codes over the binary field with code
rate $\frac{r}{r+t}$. Unfortunately, the rate does not achieve the
bound \eqref{rate-bd-2} for $t\geq 2$. For the more general case,
the $(r,t)_o$ codes, no result is known about the code rate bound
or the minimum distance bound for $t\geq 2$.

For LRCs that can exactly and locally repair $t=2$ failed nodes by
the sequential approach, it was proved in \cite{Prakash-14} that
the code rate satisfies:
\begin{align}\label{rate-bd-3}\frac{k}{n}\leq
\frac{r}{r+2}.\end{align} An upper bound for the minimum distance
of such codes was also given in \cite{Prakash-14}. However, for
$t\geq 3$, no result is known about the code rate bound or the
minimum distance bound.

\renewcommand\figurename{Fig}
\begin{figure}[htbp]
\begin{center}
\includegraphics[height=4cm]{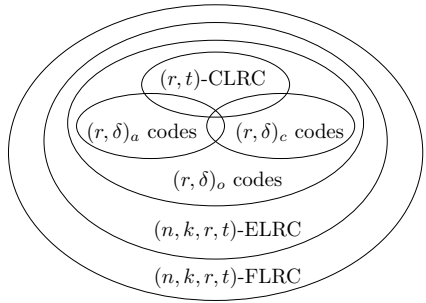}
\end{center}
\vspace{-0.2cm}\caption{Relation of the six subclasses of $[n,k]$
linear LRCs, where $\delta=t+1$. }\label{rlt-fg}
\end{figure}

\subsection{LRC with Functional Repair}
Vector codes that can locally repair single failed node with
functional repair model was considered by Hollmann et al.
\cite{Hollmann13-1}$-$\cite{Hollmann14}. Suppose $\alpha$ is the
capacity of each storage node and $\beta$ is the transport
capacity, i.e., the amount of data that can be transported from a
node contacted during the repair process. It was proved in
\cite{Hollmann14} that if $\alpha=\beta$, the code rate is upper
bounded by $\frac{r}{r+1}$, where $r$ is the repair locality.
However, the study of LRC for multiple node failures under
functional repair model is not seen in the literature.

\subsection{Our Contribution}
Given positive integers $n,k,r$ and $t$ such that $r<k<n$ and
$t\leq n-k$. We consider the problem of designing $[n,k]$ linear
codes for distributed storage systems (DSS) that satisfy the
$\bm{(r,t)}$\textbf{-\emph{Local Repair Property}}, where any $t'
(\leq t)$ simultaneously failed nodes can be locally repaired,
each with locality $r$. We consider the functional repair model
and the sequential approach for repairing multiple failed nodes.
By functional repair, we mean that the packet stored in each
newcomer is not necessarily an exact copy of the lost data but a
symbol that keep the $(r,t)$-local repair property. We call such
codes $(n,k,r,t)$-functional locally repairable code (FLRC). A
subclass of FLRC, called $(n,k,r,t)$-exact locally repairable code
(ELRC), in which the $(r,t)$-local repair property is satisfied by
exact repair, is also considered.

Clearly, codes studied in \cite{Prakash-14} are $(n,k,r,2)$-ELRC
and $(r,t)_o$ codes $($i.e., codes with overall local repair
tolerance $t)$ studied in \cite{Pamies13} are $(n,k,r,t)$-ELRC.
The relation of the six subclasses of LRCs mentioned above are
depicted in Fig. \ref{rlt-fg}.

It is easy to see that the minimum distance of an $(n,k,r,t)$-FLRC
is at least $t+1$. In this paper, our goal is to optimize the
storage overhead (equivalently, the code rate of such codes). When
$t=1$, by the result of \cite{Hollmann14}, the code rate of an
$(n,k,r,t)$-FLRC is upper bounded by $\frac{r}{r+1}$. So we focus
on the case of $t\geq 2$. Our method is to associate each
$(n,k,r,t)$-FLRC with a set of directed acyclic graphs, called
repair graph. Then by studying the structural properties of the so
called minimal repair graph $($similar to the discussion in
\cite{Fragouli06, Wentu11}$)$, we derive a lower bound of the code
length $n$. Our main results are listed as bellow:

1) We prove that for $(n,k,r,t=2)$-FLRC, the code length satisfies
$$n\geq k+\left\lceil\frac{2k}{r}\right\rceil.$$ Equivalently, the code rate
satisfies
$$\frac{k}{n}\leq\frac{r}{r+2}.$$ Note that bound \eqref{rate-bd-3}
is an upper bound of the code rate of $(n,k,r,t=2)$-ELRC. Thus,
our bound generalizes the bound \eqref{rate-bd-3} to the setting
of functional repair model.

2) We prove that for $(n,k,r,t=3)$-FLRC, the code length satisfies
$$n\geq
k+\left\lceil\frac{2k+\lceil\frac{k}{r}\rceil}{r}\right\rceil.$$
Note that codes with all-symbol $(r,\delta=4)_c$-locality is an
$(n,k,r,t=3)$-ELRC. For $t=3$, \eqref{rate-bd-2} implies that
$n\geq \frac{r+1}{r}\frac{2r+1}{2r}\frac{3r+1}{3r}k$. Moreover, we
can check that
$k+\left\lceil\frac{2k+\lceil\frac{k}{r}\rceil}{r}\right\rceil\geq
k+\frac{2k+\lceil\frac{k}{r}\rceil}{r}
\geq\frac{r+1}{r}\frac{2r+1}{2r}\frac{3r+1}{3r}k$. So our result
improves the bound \eqref{rate-bd-2} for $t=3$.

3) We give some constructions of $(n,k,r,t)$-ELRC for
$t\in\{2,3\}$ whose code length $n$ achieves the corresponding
bound of FLRC, which proves the tightness of our bounds and also
implies that there is no gap between the optimal code length of
functional LRCs and exact LRCs for some sets of parameters.
Moreover, our constructions are over the binary field, hence are
of practical interest.

\subsection{Organization}
The rest of this paper is organized as follows. In Section
\uppercase\expandafter{\romannumeral 2}, we give the basic
notations and concepts including functional locally repairable
code (FLRC), exact locally repairable code (ELRC) and repair graph
of FLRC. In section \uppercase\expandafter{\romannumeral 3}, we
prove some structural properties of the minimal repair graph of
FLRC. Lower bounds on code length of $(n,k,r,t)$-FLRC for
$t\in\{2,3\}$ are derived in Section
\uppercase\expandafter{\romannumeral 4}. Constructions of ELRC
with optimal code length is presented in Section
\uppercase\expandafter{\romannumeral 5}. The paper is concluded in
Section \uppercase\expandafter{\romannumeral 6}.

\section{Preliminary}
For any set $A$, we use $|A|$ to denote the size $($i.e., the
number of elements$)$ of $A$. A set $B$ is called an $r$-subset of
$A$ if $B\subseteq A$ and $|B|=r$. For any positive integer $n$,
we denote $[n]:=\{1,2,\cdots,n\}$. An $[n,k]$ linear code over a
field $\mathbb F$ is a $k$-dimensional subspace of $\mathbb F^n$.

Let $\mathcal C$ be an $[n,k]$ linear code over the field $\mathbb
F$. If there is no confusion in the context, we will omit the base
field $\mathbb F$ and only say that $\mathcal C$ is an $[n,k]$
linear code. A $k$-subset $S$ of $[n]$ is called an
\emph{information set} of $\mathcal C$ if for all codeword
$x=(x_1,x_2,\cdots,x_n)\in\mathcal C$ and all $i\in[n]$,
$x_i=\sum_{j\in S}a_{i,j} x_j$, where all $a_{i,j}\in\mathbb F$
and are independent of $x$. The code symbols in $\{x_j, j\in S\}$
are called \emph{information symbol} of $\mathcal C$. In contrast,
code symbols in $\{x_i, i\in [n]\backslash S\}$ are called
\emph{parity symbol} of $\mathcal C$. An $[n,k]$ linear code has
at least one information set.

For any $E\subseteq[n]$, let $\overline{E}=[n]\backslash E$ and
$\mathcal C|_{E}$ be the punctured code of $\mathcal C$ associated
with the coordinate set $E$. That is, $\mathcal C|_{E}$ is
obtained from $\mathcal C$ by deleting all code symbols in the set
$\{x_i, i\in\overline{E}\}$ for each codeword
$(x_1,x_2\cdots,x_n)\in\mathcal C$.

\subsection{Locally repairable code (LRC)}
In this subsection, we always assume that $\mathcal C$ is an
$[n,k]$ linear code over $\mathbb F$. We first present the concept
of repair set for each coordinate $i\in[n]$.

\begin{defn}\label{def-lrs}
Let $i\in[n]$ and $R\subseteq[n]\backslash\{i\}$. The subset $R$
is called an $(r,\mathcal C)$-\emph{repair set} of $i$ if $|R|\leq
r$ and $x_i=\sum_{j\in R}a_jx_j$ for all
$x=(x_1,x_2,\cdots,x_n)\in\mathcal C$, where all $a_j\in\mathbb F$
and are independent of $x$.
\end{defn}

In the following, we will omit the prefix $(r,\mathcal C)$ and say
that $R$ is a repair set of $i$ if there is no confusion in the
context.

\begin{defn}\label{r-compute} Let $E$ be a $t$-subset of $[n]$.
$\mathcal C$ is said to be $(E,r)$-repairable if there exists an
index of $E$, say $E=\{i_1,\cdots,i_t\}$, and a collection of
subsets
$$\{R_{\ell}\subseteq\overline{E}\cup\{i_1,\cdots,i_{\ell-1}\};
|R_{\ell}|\leq r, \ell\in[t]\}$$ such that for each $\ell\in[t]$,
$R_{\ell}$ is an $(r,\mathcal C)$-repair set of $i_\ell$.
\end{defn}

In this paper, we assume $r<k<n$, which means small repair
locality and at least one redundant code symbol. Moreover, if
$\mathcal C$ is $(E,r)$-repairable for some $t$-subset $E$ of
$[n]$, then we can easily see that $t\leq n-k$.

\begin{defn}\label{R-code} Let $\mathcal C'$
be an $[n,k]$ linear code over $\mathbb F~($not necessarily
different from $\mathcal C)$  and $E\subseteq[n]$. $\mathcal C'$
is said to be an $(E,r)$-repair code of $\mathcal C$ if the
following two conditions hold:
\begin{itemize}
 \item [(\romannumeral 1)] $\mathcal C|_{\overline{E}}=\mathcal
 C'|_{\overline{E}}$;
 \item [(\romannumeral 2)] $\mathcal C'$ is $(E,r)$-repairable.
\end{itemize}
\end{defn}

Consider a DSS with $n$ storage nodes where a data file is stored
as a codeword of $\mathcal C$, each node storing one code symbol.
Suppose the nodes indexed by $E$ fail. Then the symbols stored in
the live nodes form a codeword $x_{\overline{E}}$ of the punctured
code $\mathcal C|_{\overline{E}}$. If $\mathcal C'$ is an
$(E,r)$-repair code of $\mathcal C$, then $x_{\overline{E}}$ is
also a codeword of $\mathcal C'|_{\overline{E}}$. Moreover, since
$\mathcal C'$ is $(E,r)$-repairable, then we can construct a
codeword of $\mathcal C'$ from $x_{\overline{E}}$ using the
sequential approach, which form a process of functional repair.

\begin{defn}\label{f-lrc}
An $(n,k,r,t)$-\emph{functional locally repairable code (FLRC)} is
a collection of $[n,k]$ linear codes $\{\mathcal C_\lambda;
\lambda\in\Lambda\}$, where $\Lambda$ is an index set, such that
for each $\lambda\in\Lambda$ and each $E\subseteq[n]$ of size
$|E|\leq t$, there is a $\lambda'\in\Lambda$ such that $\mathcal
C_{\lambda'}$ is an $(E,r)$-repair code of $\mathcal C_\lambda$.
\end{defn}

\begin{defn}\label{e-lrc}
An $(n,k,r,t)$-\emph{exact locally repairable code (ELRC)} is an
$[n,k]$ linear code $\mathcal C$ such that for each
$E\subseteq[n]$ of size $|E|\leq t$, $\mathcal C$ is
$(E,r)$-repairable.
\end{defn}

Clearly, for any DSS with $n$ storage nodes and a data file of $k$
information symbols being stored, if the $(r,t)$-local repair
property is satisfied for functional repair model and the
sequential approach, then the coding scheme can be described as an
$(n,k,r,t)$-FLRC. Conversely, any $(n,k,r,t)$-FLRC can be used as
a coding scheme for such DSS.

Let $\{\mathcal C_\lambda; \lambda\in\Lambda\}$ be an
$(n,k,r,t)$-FLRC. Suppose $i\in[n]$ and
$\lambda_1\neq\lambda_2\in\Lambda$. It is possible that the
$(r,\mathcal C_{\lambda_1})$-repair set of $i$ is different from
the $(r,\mathcal C_{\lambda_2})$-repair set of $i$. In other
words, the repair set of the coordinate $i$ is not fixed, but
depends on the state of the system.

From Definition \ref{R-code} and \ref{e-lrc}, we can easily see
that an $[n,k]$ linear code $\mathcal C$ is an $(n,k,r,t)$-ELRC if
and only if for all $E\subseteq[n]$ of size $|E|\leq t$, $\mathcal
C$ is an $(E,r)$-repair code of itself. So an $(n,k,r,t)$-ELRC is
naturally an $(n,k,r,t)$-FLRC. Moreover, we can characterize
$(n,k,r,t)$-ELRC by a seemingly simpler condition as follows.

\begin{lem}\label{lem-ELRC}
An $[n,k]$ linear code $\mathcal C$ is an $(n,k,r,t)$-ELRC if and
only if for any $E\subseteq[n]$ of size $|E|\leq t$, there exists
an $i\in E$ such that $i$ has an $(r,\mathcal C)$-repair set
contained in $[n]\backslash E$.
\end{lem}
\begin{proof}
If $\mathcal C$ is an $(n,k,r,t)$-ELRC, then by Definition
\ref{r-compute} and \ref{e-lrc}, there exists an index of $E$, say
$E=\{i_1,\cdots,i_t\}$, such that $i_1$ has an $(r,\mathcal
C)$-repair set $R_1\subseteq\overline{E}=[n]\backslash E$.

Conversely, for any $E\subseteq[n]$ of size $|E|=t'\leq t$, by
assumption, there exists an $i_1\in E$ such that $i_1$ has an
$(r,\mathcal C)$-repair set
$R_1\subseteq\overline{E}=[n]\backslash E$. Now, let
$E_1=E\backslash\{i_1\}$. Then $|E_1|\leq t$ and by assumption,
there exists an $i_2\in E_1$ such that $i_2$ has an $(r,\mathcal
C)$-repair set $R_2\subseteq[n]\backslash
E_1=\overline{E}\cup\{i_1\}$. Similarly, we can find an $i_3\in
E\backslash\{i_1,i_2\}$ such that $i_3$ has an $(r,\mathcal
C)$-repair set $R_3\subseteq\overline{E}\cup\{i_1,i_2\}$. And so
on. Then we can index $E$ as $E=\{i_1,i_2,\cdots,i_{t'}\}$ such
that each $i_\ell$ has an $(r,\mathcal C)$-repair set
$R_\ell\subseteq\overline{E}\cup\{i_1,i_2,\cdots,i_{\ell-1}\}$.
Thus, by Definition \ref{r-compute} and \ref{e-lrc}, $\mathcal C$
is an $(n,k,r,t)$-ELRC.
\end{proof}

\subsection{Repair graph of LRC}
To derive a bound of the code length, we introduce the concepts of
repair graph and minimal repair graph of an $(n,k,r,t)$-FLRC and
investigate the structural properties of the minimal repair
graphs.

Let $G=(\mathcal V,\mathcal E)$ be a directed, acyclic graph with
node (vertex) set $\mathcal V$ and edge (arc) set $\mathcal E$.
For any $e=(u,v)\in\mathcal E$, we call $u$ the \emph{tail} of $e$
and $v$ the \emph{head} of $e$. We also call $u$ an
\emph{in-neighbor} of $v$ and $v$ an \emph{out-neighbor} of $u$.
For each $v\in \mathcal V$, let $\text{In}(v)$ and $\text{Out}(v)$
denote the set of in-neighbors and out-neighbors of $v$
respectively. If $\text{In}(v)=\emptyset$, we call $v$ a
\emph{source}. Otherwise, we call $v$ an \emph{inner node}. We use
$\text{S}(G)$ to denote the set of all sources of $G$. Moreover,
for any $V\subseteq\mathcal V$, let
\begin{align}\label{Def-Out-E}
\text{Out}(V)=\bigcup_{v\in V}\text{Out}(v)\backslash V.
\end{align} And for any
$v\in\mathcal V$, let
\begin{align}\label{Def-Out-2}
\text{Out}^2(v)=\bigcup_{u\in\text{Out}(v)}\text{Out}(u)
\backslash\text{Out}(v)
\end{align}
i.e., $\text{Out}^2(v)$ is the set of all $w\in \mathcal V$ such
that $w$ is an out-neighbor of some $u\in\text{Out}(v)$ but not an
out-neighbor of $v$.

As an example, consider the graph as depicted in Fig.
\ref{eg-fg-1}. We have $\text{Out}(3)=\{9,10\}$ and
$\text{Out}(4)=\{10,11\}$. So by \eqref{Def-Out-E},
$\text{Out}(V)=\{9,10,11\}$, where $V=\{3,4\}$. Moreover, by
\eqref{Def-Out-2}, we have $\text{Out}^2(3)=\{13,15,16\}$.

\renewcommand\figurename{Fig}
\begin{figure}[htbp]
\begin{center}
\includegraphics[height=3.6cm]{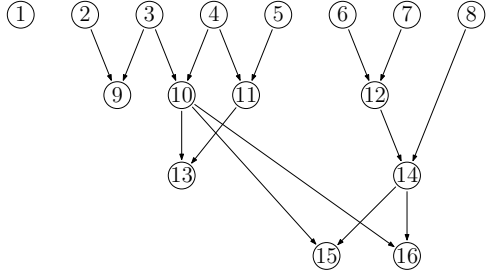}
\end{center}
\vspace{-0.2cm}\caption{An example repair graph $G_{\lambda_0}$,
where $r=2$ and $n=16$. }\label{eg-fg-1}
\end{figure}

For any linear code $\mathcal C$ with repair locality, we can
associate $\mathcal C$ with a set of graphs called \emph{repair
graph} of $\mathcal C$.

\begin{defn}\label{def-lrg}
Let $\mathcal C$ be an $[n,k]$ linear code and $G=(\mathcal
V,\mathcal E)$ be a directed, acyclic graph such that $\mathcal
V=[n]$. $G$ is called a repair graph of $\mathcal C$ if for all
inner node $i\in\mathcal V$, $\text{In}(i)$ is an $(r, \mathcal
C)$-repair set of $i$.
\end{defn}

A code $\mathcal C$ may have many repair graphs. Moreover, in
Definition \ref{def-lrg}, we do not require that $R=\text{In}(i)$
for any $(r, \mathcal C)$-repair set $R$ of $i$. Thus, it is
possible that there exists an $(r, \mathcal C)$-repair
set $R$ of $i$ such that $\text{In}(i)\neq R$. However, 
we can
always construct a repair graph $G'$ of $\mathcal C$ such that
$\text{In}(i)=R$ in $G'$.

\begin{defn}\label{def-mrg}
For any $(n,k,r,t)$-FLRC $\{\mathcal C_\lambda;
\lambda\in\Lambda\}$, let
\begin{align}\label{dlt-star} \delta^*\triangleq
\min\{|\text{S}(G_\lambda)|;\lambda\in\Lambda,G_\lambda\in\mathcal
G_\lambda\}
\end{align} where $\mathcal G_\lambda$ is the set of all repair
graphs of $\mathcal C_\lambda$. If $\lambda_0\in\Lambda$ and
$G_{\lambda_0}$ is a repair graph of $\mathcal C_{\lambda_0}$ such
that $\delta^*=|\text{S}(G_{\lambda_0})|$, then we call
$G_{\lambda_0}$ a \emph{minimal repair graph} of $\{\mathcal
C_\lambda; \lambda\in\Lambda\}$.
\end{defn}

\begin{rem}\label{ext-min-gph}
Note that for any $(n,k,r,t)$-FLRC $\{\mathcal C_\lambda;
\lambda\in\Lambda\}$,
$\{|\text{S}(G_\lambda)|;\lambda\in\Lambda,G_\lambda\in\mathcal
G_\lambda\}\subseteq[n]$ is a finite set. So by \eqref{dlt-star},
we can always find a $\lambda_0\in\Lambda$ and a repair graph
$G_{\lambda_0}$ of $\mathcal C_{\lambda_0}$ such that
$\delta^*=|\text{S}(G_{\lambda_0})|$. Thus, any $(n,k,r,t)$-FLRC
has at least one minimal repair graph.
\end{rem}

\renewcommand\figurename{Fig}
\begin{figure}[htbp]
\begin{center}
\includegraphics[height=7.0cm]{idx.1}
\end{center}
\vspace{-0.2cm}\caption{Relationship of discussions in Section
\uppercase\expandafter{\romannumeral 3} and
\uppercase\expandafter{\romannumeral 4}. }\label{TL-3-4}
\end{figure}

\section{Properties of Minimal Repair Graph}
In this section, we investigate the properties of minimal repair
graphs of $(n,k,r,t)$-FLRC, which will be used to derive a lower
bound on the code length $n$ in the next section. Our discussions
are summarized and illustrated in Fig. \ref{TL-3-4}.

In this section, we assume $\{\mathcal C_\lambda;
\lambda\in\Lambda\}$ is an $(n,k,r,t)$-FLRC and
$G_{\lambda_0}=(\mathcal V,\mathcal E)\in\mathcal G_{\lambda_0}$
is a minimal repair graph of $\{\mathcal C_\lambda;
\lambda\in\Lambda\}$, where $\lambda_0\in\Lambda$. Note that the
node set $\mathcal V=[n]$.

By Definition \ref{def-mrg} and \ref{def-lrg}, $G_{\lambda_0}$ has
$n-\delta^*$ inner nodes and each inner node of $G_{\lambda_0}$
has at most $r$ in-neighbors. So we have
\begin{align}\label{edge-num-left}
(n-\delta^*)r\geq |\mathcal E|.
\end{align}

The following lemma shows that the dimension $k$ 
is upper bounded by the number of
sources of $G_{\lambda_0}$.
\begin{lem}\label{dim-dlt}
For any $(n,k,r,t)$-FLRC $\{\mathcal C_\lambda;
\lambda\in\Lambda\}$,
\begin{align}
\label{k-leq-delta} k\leq\delta^*=|\text{S}(G_{\lambda_0})|.
\end{align}
\end{lem}
\begin{proof}
Consider an arbitrary $\lambda\in\Lambda$ and an arbitrary repair
graph $G_\lambda$ of $\mathcal C_\lambda$. By Definition
\ref{def-lrg}, $G_\lambda$ is acyclic and for each inner node $j$,
$\text{In}(j)$ is an $(r, \mathcal C)$-repair set of $j$. Then by
Definition \ref{def-lrs} and by induction, for all codeword
$x=(x_1,x_2,\cdots,x_n)\in\mathcal C_\lambda$ and all $j\in[n]$,
the $j$th code symbol $x_j$ is an $\mathbb F$-linear combination
of the symbols in $\{x_i; i\in\text{S}(G_{\lambda})\}$. So the set
$\text{S}(G_{\lambda})$ contains an information set of $\mathcal
C_\lambda$, which implies that $k\leq|\text{S}(G_{\lambda})|$.
Since $\lambda$ is an arbitrary element of $\Lambda$ and
$G_\lambda$ is an arbitrary repair graph of $\mathcal C_\lambda$,
then by Definition \ref{def-mrg}, we have
$k\leq\min\{|\text{S}(G_\lambda)|;\lambda\in\Lambda,G_\lambda\in\mathcal
G_\lambda\}=\delta^*=|\text{S}(G_{\lambda_0})|$, which proves the
lemma.
\end{proof}

The following lemma and its corollaries give some structural
properties of $G_{\lambda_0}$.
\begin{lem}\label{mrg-out}
For any $E\subseteq[n]$ of size $|E|=t'\leq t$,
\begin{align}\label{Out-geq-E}
|\text{Out}(E)|\geq |E\cap \text{S}(G_{\lambda_0})|.
\end{align}
\end{lem}
\begin{proof}
We can prove this lemma by contradiction.

By Definition \ref{f-lrc}, there is a $\lambda_1\in\Lambda$ such
that $\mathcal C_{\lambda_1}$ is an $(E,r)$-repair code of
$\mathcal C_{\lambda_0}$. By Definition \ref{r-compute}, there
exists an index of $E$, say $E=\{i_1,i_2,\cdots,i_{t'}\}$, and a
collection of subsets
$$\{R_{\ell}\subseteq\overline{E}\cup\{i_1,\cdots,i_{\ell-1}\};
|R_{\ell}|\leq r, \ell\in[t']\}$$ such that $R_{\ell}$ is an
$(r,\mathcal C_{\lambda_1})$-repair set of $i_\ell$ for each
$\ell\in[t']$. We construct a repair graph $G_{\lambda_1}$ of
$\mathcal C_{\lambda_1}$ as follows: First, for each $i\in E\cup
\text{Out}(E)$ and $j\in\text{In}(i)$, delete the edge $(j,i)$;
Then for each $i_\ell\in E$ and each $j\in R_\ell$, add a direct
edge from $j$ to $i_\ell$.

Clearly,
$\text{S}(G_{\lambda_1})=(\text{S}(G_{\lambda_0})\backslash E)\cup
\text{Out}(E)$. Here we fix the notation $\text{Out}(E)$ to be
defined in $G_{\lambda_0}$.
For each inner node $i$ of $G_{\lambda_1}$, we have the following
two cases:

Case 1: $i\in E$. Then $i=i_\ell$ for some $\ell\in[t']$ and by
the construction of $G_{\lambda_1}$, $\text{In}(i)=R_\ell$ is an
$(r,\mathcal C_{\lambda_1})$-repair set of $i$.

Case 2: $i$ is an inner node of $G_{\lambda_0}$ and
$i\notin\text{Out}(E)$. Then
$\text{In}(i)\subseteq\overline{E}=[n]\backslash E$ is an
$(r,\mathcal C_{\lambda_0})$-repair set of $i$. Moreover, since
$\mathcal C_{\lambda_1}$ is an $(E,r)$-repair code of $\mathcal
C_{\lambda_0}$, then by condition (\romannumeral 2) of Definition
\ref{R-code}, $\mathcal C_{\lambda_1}|_{\overline{E}}=\mathcal
C_{\lambda_0}|_{\overline{E}}$. So $\text{In}(i)$ is also an
$(r,\mathcal C_{\lambda_1})$-repair set of $i$.

Thus, for each inner node $i$ of $G_{\lambda_1}$, $\text{In}(i)$
is an $(r,\mathcal C_{\lambda_1})$-repair set of $i$. So
$G_{\lambda_1}$ is a repair graph of $\mathcal C_{\lambda_1}$.

Now, suppose $|\text{Out}(E)|<|E\cap \text{S}(G_{\lambda_0})|$.
Then we have
\begin{align*}
|\text{S}(G_{\lambda_1})|&=|(\text{S}(G_{\lambda_0})\backslash
E)\cup\text{Out}(E)|\\&=|(\text{S}(G_{\lambda_0})\backslash
E)|+|\text{Out}(E)|\\&=|(\text{S}(G_{\lambda_0})|-|E\cap
\text{S}(G_{\lambda_0})|+|\text{Out}(E)|\\&<|\text{S}(G_{\lambda_0})|
\end{align*}
which contradicts to Definition \ref{def-mrg}. Thus, by
contradiction, we have $|\text{Out}(E)|\geq|E\cap
\text{S}(G_{\lambda_0})|$.
\end{proof}

\begin{exam}
Let $G_{\lambda_0}$ be as in Fig. \ref{eg-fg-1} and
$G_{\lambda_0}$ be a repair graph of $\mathcal C_{\lambda_0}$ with
repair locality $r=2$. By Definition \ref{def-lrg}, $\{2,3\}$ is a
$(r,\mathcal C_{\lambda_0})$-repair set of $9$, $\{3,4\}$ is a
repair set of $10$, etc. Let $\mathcal C_{\lambda_1}$ be an
$(E=\{2,3,9\},r)$-repair code of $\mathcal C_{\lambda_0}$ such
that the $(r,\mathcal C_{\lambda_1})$-repair sets of $2,3$ and $9$
are $\{1,10\}, \{12,13\}$ and $\{11,14\}$ respectively. As in the
proof of Lemma \ref{mrg-out}, we can construct a graph
$G_{\lambda_1}$ as in Fig. \ref{eg-fg-2}. In $G_{\lambda_0}$, we
have $\text{Out}(E)=\{10\}$. In $G_{\lambda_1}$, we have
$\text{S}(G_{\lambda_1})=(\text{S}(G_{\lambda_0})\backslash
E)\cup\text{Out}(E)=(\text{S}(G_{\lambda_0})\backslash\{2,3\})\cup\{10\}$.
Moreover, we can check that $G_{\lambda_1}$ is a repair graph of
$\mathcal C_{\lambda_1}$. In fact, note that by Definition
\ref{R-code}, $\mathcal C_{\lambda_1}|_{\overline{E}}=\mathcal
C_{\lambda_0}|_{\overline{E}}$, where $\overline{E}=[n]\backslash
E$. Then $\{4,5\}$ is also an $(r,\mathcal C_{\lambda_1})$-repair
set of $11$. Similarly, $\{6,7\}$ is an $(r,\mathcal
C_{\lambda_1})$-repair set of $12$, etc. So \ref{def-lrg},
$G_{\lambda_1}$ is a repair graph of $\mathcal C_{\lambda_1}$.
\end{exam}

\renewcommand\figurename{Fig}
\begin{figure}[htbp]
\begin{center}
\includegraphics[height=3.6cm]{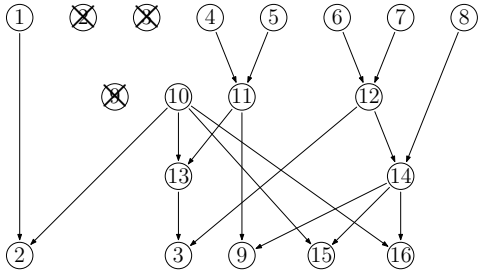}
\end{center}
\vspace{-0.2cm}\caption{The graph $G_{\lambda_1}$ obtained from
$G_{\lambda_0}$ by the process in the proof of Lemma \ref{mrg-out}
for $E=\{2,3,9\}$, where $G_{\lambda_0}$ is depicted in Fig.
\ref{eg-fg-1} and the repair sets of $2,3$ and $9$ are $\{1,10\},
\{12,13\}$ and $\{11,14\}$ respectively. 
}\label{eg-fg-2}
\end{figure}

\begin{cor}\label{mrg-cor-1}
Suppose $t\geq 3$. For any source $v$, the following hold:
\begin{itemize}
 \item [1)] $|\text{Out}(v)|\geq 1$.
 \item [2)] If $|\text{Out}(v)|=1$, then $\text{Out}^2(v)=\text{Out}(v')
 \neq\emptyset$, where $v'$ is the unique out-neighbor of $v$.
 \item [3)] If $\text{Out}(v)=\{v_1\}$ and $\text{Out}(v_1)=\{v_2\}$
 for some inner nodes $v_1$ and $v_2$, then
 $\text{Out}(v_2)\neq\emptyset$.
 \item [4)] If $\text{Out}(v)=\{v_1\}$ and $\text{Out}(v_1)=\{v_2\}$
 for some inner nodes $v_1$ and $v_2$, then $|\text{Out}(u)|\geq 2$
 for any source $u$ that belongs to $\text{In}(v_2)$.
 \item [5)] If $v$ and $w$ are two different sources and
 $|\text{Out}(v)|=|\text{Out}(w)|=1$, then the unique
 out-neighbor of $v$ is different from the unique out-neighbor of $w$.
\end{itemize}
\end{cor}
\begin{proof}
We can prove all claims by contradiction.

1) Suppose $v$ has no out-neighbor. Picking $E=\{v\}$, then
$|\text{Out}(E)|=|\emptyset|=0<|E\cap\text{S}(G_{\lambda_0})|
=|\{v\}|=1$, which contradicts to Lemma \ref{mrg-out}. (e.g., see
1) of example \ref{ex-c1}.) Thus, $v$ must have at least one
out-neighbor.

2) Suppose $\text{Out}(v')=\emptyset$. Picking $E=\{v,v'\}$, then
$|\text{Out}(E)|=|\emptyset|=0<|E\cap\text{S}(G_{\lambda_0})|
=|\{v\}|=1$, which contradicts to Lemma \ref{mrg-out}. (e.g., see
2) of example \ref{ex-c1}.) So it must be that
$\text{Out}(v')\neq\emptyset$. Since $G_{\lambda_0}$ is acyclic
and $\{v'\}=\text{Out}(v)$, then $v'\in\text{Out}(v')$. By
\eqref{Def-Out-2}, $\text{Out}^2(v)=\text{Out}(v')\neq\emptyset$.

3) Suppose $\text{Out}(v_2)=\emptyset$. Picking $E=\{v,v_1,v_2\}$,
then $|\text{Out}(E)|=|\emptyset|=0<|E\cap\text{S}(G_{\lambda_0})|
=|\{v\}|=1$, which contradicts to Lemma \ref{mrg-out}. (e.g., see
3) of example \ref{ex-c1}.) So it must be that
$\text{Out}(v_2)\neq\emptyset$.

4) Suppose $|\text{Out}(u)|<2$. Since $u\in\text{In}(v_2)$, then
$\text{Out}(u)=\{v_2\}$. Picking $E=\{v,v_1,u\}$, we have
$|\text{Out}(E)|=|\{v_2\}|=1<|E\cap\text{S}(G_{\lambda_0})|
=|\{v,u\}|=2$, which contradicts to Lemma \ref{mrg-out}. (See 4)
of example \ref{ex-c1}.) So it must be that $|\text{Out}(u)|\geq
2$.

5) Suppose $\text{Out}(v)=\text{Out}(w)=\{v_1\}$. Picking
$E=\{v,w\}$, we have
$|\text{Out}(E)|=|\{v_1\}|=1<|E\cap\text{S}(G_{\lambda_0})|
=|\{v,w\}|=2$, which contradicts to Lemma \ref{mrg-out}. (e.g.,
see 5) of example \ref{ex-c1}.) Thus, the out-neighbor of $v$ and
$w$ must be different.
\end{proof}

The following example illustrates the arguments in the proof of
Corollary \ref{mrg-cor-1}.
\begin{exam}\label{ex-c1}
For the repair graph $G_{\lambda_0}$ in Fig. \ref{eg-fg-1}, we
have the following observations:

1) Let $v=1$. Note that $\text{Out}(1)=\emptyset$. If we pick
$E=\{1\}$, then we have
$|\text{Out}(E)|=|\emptyset|=0<|E\cap\text{S}(G_{\lambda_0})|
=|\{1\}|=1$.

2) Let $v=2$ and $v'=9$. Note that $\text{Out}(9)=\emptyset$. If
we pick $E=\{2,9\}$, then
$|\text{Out}(E)|=|\emptyset|=0<|E\cap\text{S}(G_{\lambda_0})|
=|\{2\}|=1$.

3) Let $v=5, v_1=11$ and $v_2=13$. Note that
$\text{Out}(13)=\emptyset$. If we pick $E=\{5,11,13\}$, then
$|\text{Out}(E)|=|\emptyset|=0<|E\cap\text{S}(G_{\lambda_0})|
=|\{5\}|=1$.

4) Let $v=6, v_1=12, v_2=14$ and $u=8$. Note that
$|\text{Out}(8)|=1$. If we pick $E=\{6,8,12\}$, then
$|\text{Out}(E)|=|\{14\}|=1<|E\cap\text{S}(G_{\lambda_0})|
=|\{6,8\}|=2$.

5) Let $v=6, w=7$ and $v_1=12$. If we pick $E=\{6,7\}$, then
$|\text{Out}(E)|=|\{12\}|=1<|E\cap\text{S}(G_{\lambda_0})|
=|\{6,7\}|=2$.
\end{exam}

\begin{rem}\label{rem-t-1-2}
In Corollary \ref{mrg-cor-1}, 1) holds for all $t\geq 1$ and 2),
5) hold for all $t\geq 2$. In fact, in the proof of 1),
contradiction is derived from a subset $E$ of size $1$. So the
proof is valid for all $t\geq 1$. Hence, 1) holds for all $t\geq
1$. Similarly, checking the proof of 2) and 5), we can see that
they hold for all $t\geq 2$.
\end{rem}

\renewcommand\figurename{Fig}
\begin{figure}[htbp]
\begin{center}
\includegraphics[height=2.5cm]{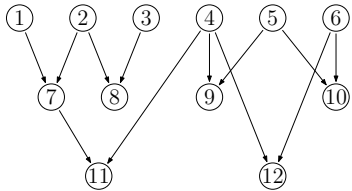}
\end{center}
\vspace{-0.2cm}\caption{An example repair graph $G$, where $n=12$
and $r=2$.}\label{eg-fg-3}
\end{figure}

\begin{cor}\label{mrg-cor-2}
Suppose $v\in\text{S}(G_{\lambda_0})$ and
$\text{Out}(v)=\{v_1,v_2\}$ for some inner nodes $v_1$ and $v_2$.
If $t\geq 3$, the following hold:
\begin{itemize}
 \item [1)] $\text{Out}(v_1)\neq\emptyset$ or
 $\text{Out}(v_2)\neq\emptyset$.
 \item [2)] If $v_1=\text{Out}(u)$ for some source $u$, then
 $\text{Out}(v_2)\neq\emptyset$.
 \item [3)] If $v_1=\text{Out}(u)$ for some source $u$, then
 $|\text{Out}(w)|\geq 2$ for any source $w$ that belongs to
 $\text{In}(v_2)$.
\end{itemize}
\end{cor}
\begin{proof}
We can prove all claims by contradiction.

1) Suppose $\text{Out}(v_1)=\emptyset$ and
$\text{Out}(v_2)=\emptyset$. Picking $E=\{v,v_1,v_2\}$, we have
$|\text{Out}(E)|=|\emptyset|=0<|E\cap\text{S}(G_{\lambda_0})|
=|\{v\}|=1$, which contradicts to Lemma \ref{mrg-out}. (See 1) of
example \ref{ex-c2}.) So it must be that
$\text{Out}(v_1)\neq\emptyset$ or $\text{Out}(v_2)\neq\emptyset$.

2) Suppose $\text{Out}(v_2)=\emptyset$. Picking $E=\{u,v,v_2\}$,
we have
$|\text{Out}(E)|=|\{v_1\}|=1<|E\cap\text{S}(G_{\lambda_0})|
=|\{v,u\}|=2$, which contradicts to Lemma \ref{mrg-out}. (e.g.,
see 2) of example \ref{ex-c2}.) So it must be that
$\text{Out}(v_2)\neq\emptyset$.

3) Suppose $w\in\text{In}(v_2)$ is a source and $|\text{Out}(w)|<
2$. Then $\text{Out}(w)=\{v_2\}$. Picking $E=\{u,v,w\}$, we have
$|\text{Out}(E)|=|\{v_1,v_2\}|=2<|E\cap\text{S}(G_{\lambda_0})|
=|\{u,v,w\}|=3$, which contradicts to Lemma \ref{mrg-out}. (e.g.,
see 3) of example \ref{ex-c2}.) So it must be that
$|\text{Out}(w)|\geq 2$.
\end{proof}

The following example illustrates the arguments in the proof of
Corollary \ref{mrg-cor-2}.
\begin{exam}\label{ex-c2}
For the repair graph $G$ in Fig. \ref{eg-fg-3}, we have the
following observations:

1) Let $v=5, v_1=9$ and $v_2=10$. Note that
$\text{Out}(9)=\text{Out}(10)=\emptyset$. If we pick
$E=\{5,9,10\}$, then we have
$|\text{Out}(E)|=|\emptyset|=0<|E\cap\text{S}(G)| =|\{5\}|=1$.

2) Let $v=2, v_1=7, v_2=8$ and $u=1$. Note that
$\text{Out}(8)=\emptyset$. If we pick $E=\{1,2,8\}$, then
$|\text{Out}(E)|=|\{7\}|=1<|E\cap\text{S}(G_{\lambda_0})|
=|\{1,2\}|=2$.

3) Let $v=2, v_1=7, v_2=8, u=1$ and $w=3$. Note that
$|\text{Out}(3)|=1$. If we pick $E=\{1,2,3\}$, then
$|\text{Out}(E)|=|\{7,8\}|=2<|E\cap\text{S}(G_{\lambda_0})|
=|\{1,2,3\}|=3$.
\end{exam}

\section{Bound of Code Length}
In this section, we will prove a lower bound on the code length
$n$ for $(n,k,r,t)$-FLRC with $t\in\{2,3\}$.

\subsection{Code Length for $(n,k,r,2)$-FLRC}
The following theorem gives a lower bound on the code length of
$(n,k,r,2)$-FLRC.
\begin{thm}\label{bnd-t-1-2}
For $(n,k,r,2)$-FLRC, we have
\begin{align}\label{eq-bnd-t-1-2}
n\geq k+\left\lceil\frac{2k}{r}\right\rceil.
\end{align}
\end{thm}
\begin{proof}
Suppose $\{\mathcal C_\lambda; \lambda\in\Lambda\}$ is an
$(n,k,r,2)$-FLRC and $G_{\lambda_0}=(\mathcal V,\mathcal E)$ is a
minimal repair graph of $\{\mathcal C_\lambda;
\lambda\in\Lambda\}$, where $\lambda_0\in\Lambda$, $\mathcal
V=[n]$ is the node set of $G_{\lambda_0}$ and $\mathcal E$ is the
edge set of $G_{\lambda_0}$. We first prove $n\geq
\delta^*+\frac{2\delta^*}{r}$, where
$\delta^*=|\text{S}(G_{\lambda_0})|$.

By Remark \ref{rem-t-1-2} and 1) of Corollary \ref{mrg-cor-1},
each source of $G_{\lambda_0}$ has at least one out-neighbor. Let
$\mathcal E_{\text{red}}$ be the set of all edge $e$ such that the
tail of $e$ is a source. We call each edge in $\mathcal
E_{\text{red}}$ a \emph{red edge}. Let $A$ be the set of all
source $v$ such that $v$ has only one out-neighbor. Then the
number of all red edges is $|\mathcal E_{\text{red}}|\geq
|A|+2(|\text{S}(G_{\lambda_0})\backslash
A|)=|A|+2(|\text{S}(G_{\lambda_0})|-|A|)=2|\text{S}(G_{\lambda_0})|-|A|
=2\delta^*-|A|$. Thus, we have
\begin{align}\label{eq-bnd-t-1-2-1} |\mathcal E_{\text{red}}|\geq
2\delta^*-|A|.
\end{align}

For each $v\in A$, since $v$ has only one out-neighbor, by Remark
\ref{rem-t-1-2} and 2) of Corollary \ref{mrg-cor-1},
$\text{Out}^2(v)=\text{Out}(v')\neq\emptyset$, where $v'$ is the
unique out-neighbor of $v$. Let $\mathcal E_{\text{green}}(v)$ be
the set of all edges whose tail is $v'$. Then $\mathcal
E_{\text{green}}(v)\neq\emptyset$. Let $\mathcal
E_{\text{green}}=\bigcup_{v\in A}\mathcal E_{\text{green}}(v)$. We
call each edge in $\mathcal E_{\text{green}}$ a \emph{green edge}.
For any two different $v_1,v_2\in A$, let $v_1',v_2'$ be the
unique out-neighbor of $v_1,v_2$ respectively. By Remark
\ref{rem-t-1-2} and 5) of Corollary \ref{mrg-cor-1}, $v_1'\neq
v_2'$. So we have $\mathcal E_{\text{green}}(v_1)\cap \mathcal
E_{\text{green}}(v_2)=\emptyset$. Thus, the number of all green
edges is $|\mathcal E_{\text{green}}|=|\bigcup_{v\in A}\mathcal
E_{\text{green}}(v)|=\sum_{v\in A}|\mathcal
E_{\text{green}}(v)|\geq |A|$, i.e.,
\begin{align}\label{eq-bnd-t-1-2-2} |\mathcal E_{\text{green}}|\geq |A|.
\end{align}
Clearly, $\mathcal E_{\text{red}}\cap \mathcal
E_{\text{green}}=\emptyset$. Then by \eqref{eq-bnd-t-1-2-1} and
\eqref{eq-bnd-t-1-2-2}, we have $$|\mathcal E|\geq|\mathcal
E_{\text{red}}\cup \mathcal E_{\text{green}}|=|\mathcal
E_{\text{red}}|+|\mathcal E_{\text{green}}|\geq 2\delta^*.$$ On
the other hand, by \eqref{edge-num-left}, we have
$$(n-\delta^*)r\geq |\mathcal E|.$$ Thus, we have
$(n-\delta^*)r\geq 2\delta^*$, which implies that
$nr\geq\delta^*(r+2)$. So $n\geq
\frac{\delta^*(r+2)}{r}=\delta^*+\frac{2\delta^*}{r}$.

By Lemma \ref{dim-dlt}, $k\leq\delta^*=|\text{S}(G_{\lambda_0})|$.
So $n\geq\delta^*+\frac{2\delta^*}{r}\geq k+\frac{2k}{r}$.
Moreover, since $n$ is an positive integer, then we have $n\geq
k+\left\lceil\frac{2k}{r}\right\rceil$, which proves
\eqref{eq-bnd-t-1-2}.
\end{proof}

In \cite{Prakash-14}, it was proved that the code rate of an
$(n,k,r,2)$-ELRC satisfies bound \eqref{rate-bd-3}. Note that
\eqref{eq-bnd-t-1-2} also implies $\frac{k}{n}\leq\frac{r}{r+2}$.
So our result generalizes bound \eqref{rate-bd-3} to
$(n,k,r,2)$-FLRC.

\subsection{Code Length for $(n,k,r,3)$-FLRC}
The following theorem gives a lower bound on the code length of
$(n,k,r,3)$-FLRC.
\begin{thm}\label{bnd-t-3}
For $(n,k,r,3)$-FLRC, we have
\begin{align}\label{eq-bnd-t-3}
n\geq
k+\left\lceil\frac{2k+\lceil\frac{k}{r}\rceil}{r}\right\rceil.
\end{align}
\end{thm}

Before proving Theorem \ref{bnd-t-3}, we first prove the following
Lemma \ref{mrg-edge-num}. In the rest of this subsection, we
always assume $\{\mathcal C_\lambda; \lambda\in\Lambda\}$ is an
$(n,k,r,3)$-FLRC and $G_{\lambda_0}=(\mathcal V,\mathcal E)$ is a
minimal repair graph of $\{\mathcal C_\lambda;
\lambda\in\Lambda\}$, where $\lambda_0\in\Lambda$, $\mathcal
V=[n]$ is the node set of $G_{\lambda_0}$ and $\mathcal E$ is the
edge set of $G_{\lambda_0}$. Then
$\delta^*=|\text{S}(G_{\lambda_0})|$, where $\delta^*$ is defined
by \eqref{dlt-star}.
\begin{lem}\label{mrg-edge-num}
For $(n,k,r,3)$-FLRC, we have
\begin{align}\label{edge-num} (n-\delta^*)r\geq |\mathcal E|
\geq 2\delta^*+\left\lceil\frac{\delta^*}{r}
\right\rceil.\end{align}
\end{lem}
\begin{proof}
By \eqref{edge-num-left}, we have $(n-\delta^*)r\geq |\mathcal
E|$, which proves the first inequality of \eqref{edge-num}. So we
only need to prove the second inequality of \eqref{edge-num}. To
do this, we will divide the source set $\text{S}(G_{\lambda_0})$
and the edge set $\mathcal E$ into mutually disjoint subsets.

We can divide the source set $\text{S}(G_{\lambda_0})$ into four
subsets $A, B, C_1$ and $C_2$ as follows:
\begin{align}\label{divide-a}
A=\{v\in\text{S}(G_{\lambda_0}); |\text{Out}(v)|\geq
3\},\end{align}
\begin{align}\label{divide-b}B=\{v\in\text{S}(G_{\lambda_0});
|\text{Out}(v)|=2\},\end{align}
\begin{align}
\label{divide-c1}C_1=\{v\in\text{S}(G_{\lambda_0});
|\text{Out}(v)|=1 \text{~and~} |\text{Out}^2(v)|=1\}\end{align}
and
\begin{align}\label{divide-c2}C_2=\{v\in\text{S}(G_{\lambda_0});
|\text{Out}(v)|=1
\text{~and~} |\text{Out}^2(v)|\geq 2\}.\end{align} Clearly, $A, B,
C_1$ and $C_2$ are mutually disjoint. Moreover, by 1), 2) of
Corollary \ref{mrg-cor-1}, $\text{S}(G_{\lambda_0})=A\cup B\cup
C_1\cup C_2$. Hence, \begin{align}\label{num-source}
\delta^*=|\text{S}(G_{\lambda_0})|=|A|+|B|+|C_1|+|C_2|.
\end{align}

We can divide the edge set $\mathcal E$ into three subsets as
follows.

Firstly, an edge is called a \emph{red edge} if its tail is a
source. For each $v\in\text{S}(G_{\lambda_0})$, let $\mathcal
E_{\text{red}}(v)$ be the set of all red edges whose tail is $v$
and $\mathcal
E_{\text{red}}=\bigcup_{v\in\text{S}(G_{\lambda_0})}\mathcal
E_{\text{red}}(v)$ be the set of all red edges. Clearly,
$|\mathcal E_{\text{red}}(v)|=|\text{Out}(v)|$ and $\mathcal
E_{\text{red}}(w)\cap\mathcal E_{\text{red}}(v)=\emptyset$ for any
source $w\neq v$. So by \eqref{divide-a}$-$\eqref{divide-c2}, we
have
\begin{align}\label{num-red-edge}
|\mathcal
E_{\text{red}}|=\sum_{v\in\text{S}(G_{\lambda_0})}|\text{Out}(v)|\geq
3|A|+2|B|+|C_1|+|C_2|.
\end{align}

Secondly, an edge is called a \emph{green edge} if its tail is the
unique out-neighbor of some source in $C_1\cup C_2$. For each
$v\in
C_1\cup C_2$, 
let $\mathcal E_{\text{green}}(v)$ be the set of all green edges
whose tail is the unique out-neighbor of $v$ and $\mathcal
E_{\text{green}}=\bigcup_{v\in C_1\cup C_2}\mathcal
E_{\text{green}}(v)$ be the set of all green edges. Note that by
2) of Corollary \ref{mrg-cor-1},
$\text{Out}^2(v)=\text{Out}(v')\neq\emptyset$, where $v'$ is the
unique out-neighbor of $v$. Then $|\mathcal
E_{\text{green}}(v)|=|\text{Out}^2(v)|$. Moreover, if $v,w\in
C_1\cup C_2$ are different, then by 5) of Corollary
\ref{mrg-cor-1}, their out-neighbors are different. So $\mathcal
E_{\text{green}}(v)\cap\mathcal E_{\text{green}}(w)=\emptyset$.
Hence, by \eqref{divide-c1} and \eqref{divide-c2}, we have
\begin{align}\label{num-green-edge}
|\mathcal E_{\text{green}}|=\sum_{v\in C_1\cup
C_2}|\text{Out}^2(v)|\geq|C_1|+2|C_2|.
\end{align}

Thirdly, suppose $v\in B\cup C_1$ and $e\in\mathcal E$ such that
$e$ is neither a red edge nor a green edge. Then $e$ is called a
\emph{blue edge} belonging to $v$ if one of the following
conditions hold:
\begin{itemize}
 \item[(a)] $v\in B$ and the tail of $e$ belongs to $\text{Out}(v)$.
 \item[(b)] $v\in C_1$ and the tail of $e$ belongs to $\text{Out}^2(v)$.
\end{itemize}
Let $\mathcal E_{\text{blue}}(v)$ denote the set of all blue edges
belonging to $v$ and $\mathcal E_{\text{blue}}=\bigcup_{v\in B\cup
C_1}\mathcal E_{\text{blue}}(v)$. We have the following claim 1,
whose proof is given in Appendix A.

\textbf{Claim 1}: The number of blue edges is bounded by
\begin{align}\label{num-blue-edge}
|\mathcal E_{\text{blue}}|\geq\frac{|B|+|C_1|}{r}.
\end{align}

Clearly, $\mathcal E_{\text{red}},\mathcal E_{\text{green}}$ and
$\mathcal E_{\text{blue}}$ are mutually disjoint. Then by
\eqref{num-source}-\eqref{num-blue-edge}, we have
\begin{align*}
|\mathcal E|&\geq|\mathcal E_{\text{red}}|+|\mathcal
E_{\text{green}}|+|\mathcal E_{\text{blue}}|\\&\geq
(3|A|+2|B|+|C_1|+|C_2|)\\& ~ ~ ~ +(|C_1|+2|C_2|)+\frac{|B|+|C_1|}{r}\\
&=
2(|A|+|B|+|C_1|+|C_2|)\\& ~ ~ ~ +(|A|+|C_2|+\frac{|B|+|C_1|}{r})\\
&=
2\delta^*+\frac{r|A|+r|C_2|+|B|+|C_1|}{r}\\
&\geq
2\delta^*+\frac{|A|+|C_2|+|B|+|C_1|}{r}\\
&=2\delta^*+\frac{\delta^*}{r}.
\end{align*}
Note that $|\mathcal E|$ is an integer. Then we have $|\mathcal
E|\geq 2\delta^*+\left\lceil\frac{\delta^*}{r}\right\rceil$, which
proves the second inequality of \eqref{edge-num}.

By the above discussion, we proved \eqref{edge-num}, which in turn
proves Lemma \ref{mrg-edge-num}.
\end{proof}

To help the reader to understand the proof of Lemma
\ref{mrg-edge-num}, we give an example as follows.
\begin{exam}
Consider the graph in Fig. \ref{eg-fg-4}. Using the notations
defined in the proof of Lemma \ref{mrg-edge-num}, we have
$A=\{2,4,7\}$, $B=\{3,6\}$, $C_1=\{1\}$ and $C_2=\{5\}$.

It is easy to find all red edges. We can also easily find that
$\mathcal E_{\text{green}}(1)=\{(8,11)\}$ and $\mathcal
E_{\text{green}}(5)=\{(10,12), (10,13)\}$.

Since $1\in C_1$ and $11\in\text{Out}^2(1)$, then
$(11,14)\in\mathcal E_{\text{blue}}(1)$; Since
$11\in\text{Out}(6)$ and $6\in B$, then $(11,14)\in\mathcal
E_{\text{blue}}(6)$; Since $3\in B$ and $9\in\text{Out}(3)$, then
$(9,11)\in\mathcal E_{\text{blue}}(3)$. We can further check that
$\mathcal E_{\text{blue}}(1)=\mathcal
E_{\text{blue}}(6)=\{(11,14)\}$ and $\mathcal
E_{\text{blue}}(3)=\{(9,11)\}$.
\end{exam}

\renewcommand\figurename{Fig}
\begin{figure}[htbp]
\begin{center}
\includegraphics[height=3.6cm]{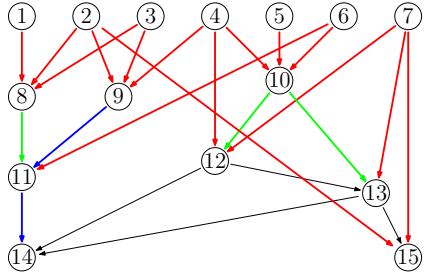}
\end{center}
\vspace{-0.2cm}\caption{An example of partitioning the edge set of
minimal repair graph: The red (resp. green, blue) edges are
colored by red (resp. green, blue).}\label{eg-fg-4}
\end{figure}

Now, using Lemma \ref{mrg-edge-num} and Lemma \ref{dim-dlt}, we
can give a simple proof of Theorem \ref{bnd-t-3}.
\begin{proof}[Proof of Theorem \ref{bnd-t-3}]
By Lemma \ref{mrg-edge-num}, we have
\begin{align*}(n-\delta^*)r\geq |\mathcal E| \geq
2\delta^*+\left\lceil\frac{\delta^*}{r} \right\rceil.\end{align*}
So
\begin{align*}(n-\delta^*)r\geq
2\delta^*+\left\lceil\frac{\delta^*}{r} \right\rceil.\end{align*}
Solving $n$ from the above equation, we can obtain
\begin{align}\label{n-dlt-r} n\geq\delta^*+\frac{2\delta^*+
\left\lceil\frac{\delta^*}{r} \right\rceil}{r} .\end{align} By
Lemma \ref{dim-dlt}, we have $\delta^*\geq k$. So
\begin{align}\label{dlt-r-k} \delta^*+\frac{2\delta^*+
\left\lceil\frac{\delta^*}{r}\right\rceil}{r}\geq k+\frac{2k+
\left\lceil\frac{k}{r} \right\rceil}{r} .\end{align} From
\eqref{n-dlt-r} and \eqref{dlt-r-k}, we have
\begin{align*} n\geq k+\frac{2k+
\left\lceil\frac{k}{r} \right\rceil}{r}.\end{align*} Since $n$ is
a positive integer, then we have $n\geq k+\left\lceil\frac{2k+
\left\lceil\frac{k}{r} \right\rceil}{r}\right\rceil$, which proves
Theorem \ref{bnd-t-3}.
\end{proof}

\renewcommand\figurename{Fig}
\begin{figure}[htbp]
\begin{center}
\includegraphics[height=6cm]{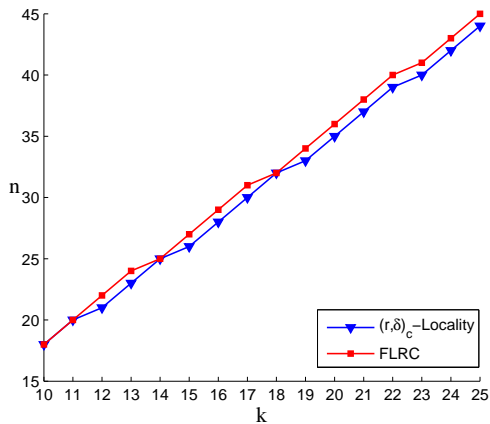}
\end{center}
\vspace{-0.2cm}\caption{Comparison of the code length bounds for
$t=r=3$ and $\delta=t+1=4$. }\label{bj-fg-1}
\end{figure}

We next show that the bound \eqref{eq-bnd-t-3} improves the bound
\eqref{rate-bd-2} for codes with all-symbol $(r,4)_c$-locality.
Note that for such codes, the bound \eqref{rate-bd-2} is
equivalent to
\begin{align}\label{bj-1}
n\geq \frac{r+1}{r}\frac{2r+1}{2r}\frac{3r+1}{3r}k.\end{align}
Also note that codes with all-symbol $(r,4)_c$-locality are
$(n,k,r,3)$-ELRC. Then by \eqref{eq-bnd-t-3}, we have
\begin{align}\label{bj-2}
n\geq
k+\left\lceil\frac{2k+\lceil\frac{k}{r}\rceil}{r}\right\rceil\geq
k+\frac{2k+\lceil\frac{k}{r}\rceil}{r}.\end{align} It is easy to
check that
\begin{align*}
&\left(k+\frac{2k+\lceil\frac{k}{r}\rceil}{r}\right)
-\left(\frac{r+1}{r}\frac{2r+1}{2r}\frac{3r+1}{3r}k\right)
\\&=\frac{1}{r}\left(\left\lceil\frac{k}{r}\right\rceil-\frac{k}{r}\right)+
\frac{k}{6r}\left(1-\frac{1}{r^2}\right)\\&\geq 0.\end{align*} So
\eqref{bj-2} is an improvement of \eqref{bj-1}.

An illustration of the gap between the bounds \eqref{eq-bnd-t-3}
and \eqref{rate-bd-2} for the parameters $t=r=3$ is given in Fig.
\ref{bj-fg-1}, from which we can see that \eqref{eq-bnd-t-3} is
tighter than \eqref{rate-bd-2} for $t=3$.

\renewcommand\figurename{Fig}
\begin{figure}[htbp]
\begin{center}
\includegraphics[height=7.0cm]{idx.2}
\end{center}
\vspace{-0.2cm}\caption{Relationship of discussions in Section
\uppercase\expandafter{\romannumeral 5}.}\label{TL-5}
\end{figure}

\section{Code construction}
In this section, we give some constructions of $(n,k,r,2)$-ELRCs
and $(n,k,r,3)$-ELRCs whose length $n$ achieve the bounds
\eqref{eq-bnd-t-1-2} and \eqref{eq-bnd-t-3} respectively. We call
such codes optimal $(n,k,r,2)$-ELRC and optimal $(n,k,r,3)$-ELRC
respectively. By these constructions, we prove the tightness of
the bound \eqref{eq-bnd-t-1-2} and \eqref{eq-bnd-t-3}. Moreover
interestingly, our results show that for some sets of parameters,
exact LRCs is sufficient to achieve the optimal code length of
functional LRCs. Our discussions are summarized and illustrated in
Fig. \ref{TL-5}.

We begin with a lemma that gives a method to construct subsets of
$[n]$ that can be used to construct repair set for LRC.
\begin{lem}\label{M-2-S}
Let $\mathcal L=\{C_1,\cdots,C_N\}$ be a collection of pairwise
disjoint subsets of $[n]$ and $(r_1,r_2,\cdots,r_K)$ be a
$K$-tuple of positive integers such that
$\sum_{i=1}^N|C_i|=\sum_{i=1}^Kr_i$. Let $M$ be a $K\times N$
binary matrix such that for each $i\in [K]$ and each $j\in[N]$,
the sum of the $i$th row is $r_i$ and the sum of the $j$th column
is $|C_j|$. Then there exists a collection $\{B_1,\cdots,B_K\}$ of
subsets of $\bigcup_{j=1}^NC_j$ such that:
\begin{itemize}
 \item [(\romannumeral 1)] $B_1,\cdots,B_K$ are pairwise disjoint and
 $\bigcup_{i=1}^KB_i=\bigcup_{j=1}^NC_j$;
 \vspace{0.2cm}\item [(\romannumeral 2)] $|B_i|=r_i$ for all $i\in[K]$;
 \vspace{0.1cm}\item [(\romannumeral 3)] $|B_i\cap C_j|\leq 1$ for all $i\in[K]$
 and $j\in[N]$.
\end{itemize}
\end{lem}
\begin{proof}
For each $j\in[N]$, since the sum of the $j$th column of $M$ is
$|C_j|$, we can replace the ones of the $j$th column by elements
of $C_j$ such that each element of $C_j$ appears exactly once.
Denote the resulted matrix by $M'$. Now for each $i\in[K]$, let
$B_i$ be the elements of the $i$th row of $M'$ except the zeros.

Since $C_1,\cdots,C_N$ are pairwise disjoint and for each
$j\in[N]$, each element of $C_j$ appears exactly once in the $j$th
column of $M'$, then each element of $\bigcup_{j=1}^NC_j$ appears
exactly once in $M'$, which implies conditions (\romannumeral 1)
and (\romannumeral 3). Moreover, since the sum of the $i$th row of
$M$ is $r_i$, then $|B_i|=r_i$ for all $i\in[K]$. So condition
(\romannumeral 2) is satisfied.
\end{proof}

We give an example in the below to demonstrate the construction
method used in the proof Lemma \ref{M-2-S}.
\begin{exam}\label{ex-t-2}
Let $C_1=\{1,2,3,4,5\}$, $C_2=\{6,7,8,9,10\}$,
$C_3=\{11,12,13,14,15\}$, $C_4=\{16,17,18,19,20\}$,
$C_5=\{22,23,24,25\}$, $C_6=\{27,28,29,30\}$ and
$C_7=\{31,32,33\}$. Let $r_1=\cdots=r_5=5$ and $r_6=r_7=3$. Then
we have $\sum_{i=1}^7|C_i|=31=\sum_{i=1}^7r_i$. Let
\begin{eqnarray*} M=\left(\begin{array}{ccccccc}
1 & 1 & 1 & 1 & 1 & 0 & 0\\
0 & 1 & 1 & 1 & 0 & 1 & 1\\
1 & 1 & 1 & 1 & 1 & 0 & 0\\
1 & 0 & 1 & 1 & 0 & 1 & 1\\
1 & 1 & 0 & 0 & 1 & 1 & 1\\
1 & 1 & 1 & 0 & 0 & 0 & 0\\
0 & 0 & 0 & 1 & 1 & 1 & 0\\
\end{array}\right).
\end{eqnarray*}
We can check that for each $i,j\in\{1,2,\cdots,7\}$, the sum of
the $i$th row is $r_i$ and the sum of the $j$th column is $|C_j|$.
Replacing the ones of the $j$th column of $M$ by elements of
$C_j$, we obtain
\begin{eqnarray*} M'=\left(\begin{array}{ccccccc}
1 & 6 & 11 & 16 & 22 & 0 & 0\\
0 & 7 & 12 & 17 & 0 & 27 & 31\\
2 & 8 & 13 & 18 & 23 & 0 & 0\\
3 & 0 & 14 & 19 & 0 & 28 & 32\\
4 & 9 & 0 & 0 & 24 & 29 & 33\\
5 & 10 & 15 & 0 & 0 & 0 & 0\\
0 & 0 & 0 & 20 & 25 & 30 & 0\\
\end{array}\right)
\end{eqnarray*}
From $M'$, we can obtain subsets $B_1=\{1,6,11,16,22\}$,
$B_2=\{7,12,17,27,31\}$, $B_3=\{2,8,13,18,23\}$,
$B_4=\{3,14,19,28,32\}$, $B_5=\{4,9,24,29,33\}$, $B_6=\{5,10,15\}$
and $B_7=\{20,25,30\}$. It is easy to check that conditions
(\romannumeral 1)$-$(\romannumeral 3) of Lemma \ref{M-2-S} are
satisfied.
\end{exam}

\begin{cor}\label{subset-eq-c}
Let $\mathcal L=\{C_1,\cdots,C_N\}$ be a collection of pairwise
disjoint $\delta$-subsets of $[n]$ and $\vec{r}=(r_1,\cdots,r_K)$
be a $K$-tuple of positive integers such that
$\sum_{i=1}^Kr_i=\delta N$ and $r_i\leq|\mathcal L|=N$ for all
$i\in[K]$. Then there exists a collection $\{B_1,\cdots,B_K\}$ of
subsets of $\bigcup_{j=1}^NC_j$ such that:
\begin{itemize}
 \item [(\romannumeral 1)] $B_1,\cdots,B_K$ are pairwise disjoint and
 $\bigcup_{i=1}^KB_i=\bigcup_{j=1}^NC_j$;
 \vspace{0.2cm}\item [(\romannumeral 2)] $|B_i|=r_i$ for all $i\in[K]$;
 \vspace{0.1cm}\item [(\romannumeral 3)] $|B_i\cap C_j|\leq 1$ for all $i\in[K]$
 and $j\in[N]$.
\end{itemize}
\end{cor}
\begin{proof}
Since $\sum_{i=1}^Kr_i=\delta N$ and $r_i\leq N$ for all
$i\in[K]$, using the Gale-Ryser Theorem $($see Manfred
\cite{Manfred}$)$, we can construct a $K\times N$ binary matrix
$M$ such that for each $i\in[K]$ and each $j\in[N]$, the sum of
the $i$th row of $M$ is $r_i$ and the sum of the $j$th column of
$M$ is $\delta=|C_j|$. By Lemma \ref{M-2-S}, there exists a
collection $\{B_1,\cdots,B_K\}$ of subsets of $\bigcup_{j=1}^NC_j$
that satisfies the conditions (\romannumeral 1)$-$(\romannumeral
3).
\end{proof}

The following two lemmas give a sufficient condition of
$(n,k,r,2)$-ELRC and $(n,k,r,3)$-ELRC respectively.

\begin{lem}\label{be-2-elrc}
Let $\mathcal C$ be an $[n,k]$ linear code and $[n]=S\cup T$ such
that $S\cap T=\emptyset$. Then $\mathcal C$ is an $(n,k,r,2)$-ELRC
if the following two conditions hold:
\begin{itemize}
 \item [(\romannumeral 1)] Each $i\in S$ has two disjoint
 $(r,\mathcal C)$-repair sets;
 \item [(\romannumeral 2)] Each $i\in T$ has an $(r,\mathcal C)$-repair
 set $R\subseteq S$.
\end{itemize}
\end{lem}
\begin{proof}
We will prove that for any $E\subseteq[n]$ of size $|E|\leq 2$,
there exists an $i\subseteq E$ such that $i$ has an $(r,\mathcal
C)$-repair set $R\subseteq [n]\backslash E$. We have the following
two cases:

Case 1: $E\cap S=\emptyset$. Then $E\subseteq T$ and by condition
(\romannumeral 2) each $i\in E$ has an $(r,\mathcal C)$-repair set
$R\subseteq S\subseteq[n]\backslash E$.

Case 2: $E\cap S\neq\emptyset$. Suppose $i\in E\cap S$. By
condition (\romannumeral 1), $i$ has two disjoint $(r,\mathcal
C)$-repair sets, say $R_1$ and $R_2$. Note that $|E|\leq 2$ and
$i\notin R_1\cup R_2$, then either $E\cap R_1=\emptyset$ or $E\cap
R_2=\emptyset$. Without loss of generality, assume $E\cap
R_1=\emptyset$. Then we have $R_1\subseteq [n]\backslash E$.

Thus, we can always find an $i\in E$ that has an $(r,\mathcal
C)$-repair set $R\subseteq [n]\backslash E$. By Lemma
\ref{lem-ELRC}, $\mathcal C$ is an $(n,k,r,2)$-ELRC.
\end{proof}

\begin{lem}\label{be-3-elrc}
Let $\mathcal C$ be an $[n,k]$ linear code and $[n]=S\cup T$ such
that $S\cap T=\emptyset$. Then $\mathcal C$ is an $(n,k,r,3)$-ELRC
if the following two conditions hold:
\begin{itemize}
 \item [(\romannumeral 1)] Each $i\in S$ has two disjoint
 $(r,\mathcal C)$-repair sets, say $R_1$ and
 $R_2$, such that each $j\in R_1$ has an $(r,\mathcal C)$-repair set
 $R\cap (R_2\cup\{i\})=\emptyset$;
 \item [(\romannumeral 2)] Each $i\in T$ has an $(r,\mathcal C)$-repair
 set $R\subseteq S$;
\end{itemize}
\end{lem}
\begin{proof}
For any $E\subseteq[n]$ of size $|E|\leq 3$, similar to the proof
of Lemma \ref{be-2-elrc}, we have the following two cases:

Case 1: $E\cap S=\emptyset$. Then $E\subseteq T$ and by condition
(\romannumeral 2) each $i\in E$ has an $(r,\mathcal C)$-repair set
$R\subseteq S\subseteq[n]\backslash E$.

Case 2: $E\cap S\neq\emptyset$. Let $i\in E\cap S$. By condition
(\romannumeral 1), $i$ has two disjoint $(r,\mathcal C)$-repair
sets, say $R_1$ and $R_2$, such that each $j\in R_1$ has an
$(r,\mathcal C)$-repair set $R\cap (R_2\cup\{i\})=\emptyset$. Then
we have the following two subcases:

Case 2.1: $E\cap R_1=\emptyset$ or $E\cap R_2=\emptyset$. If
$E\cap R_1=\emptyset$, then $R_1\subseteq [n]\backslash E$; If
$E\cap R_2=\emptyset$, then $R_2\subseteq [n]\backslash E$. So in
this subcase, $i$ has an $(r,\mathcal C)$-repair set contained in
$[n]\backslash E$.

Case 2.2: $E\cap R_1\neq\emptyset$ and $E\cap R_2\neq\emptyset$.
Assume $j\in E\cap R_1$ and $j'\in E\cap R_2$. Then by condition
(\romannumeral 1), $j$ has an $(r,\mathcal C)$-repair set $R\cap
(R_2\cup\{i\})=\emptyset$. So
\begin{align}\label{R-cap}
R\cap (R_2\cup\{i,j\})=\emptyset. \end{align} On the other hand,
since $R_1\cap R_2=\emptyset$ and $|E|\leq 3$, then $j\neq j'$ and
\begin{align}\label{E-eq}
E=\{i,j,j'\}\subseteq R_2\cup\{i,j\}. \end{align} Combining
\eqref{E-eq} and \eqref{R-cap}, we have $R\subseteq [n]\backslash
E$. So in this subcase, $j\in E$ has an $(r,\mathcal C)$-repair
set $R\subseteq [n]\backslash E$.

Thus, we can find an element of $E$ that has an $(r,\mathcal
C)$-repair set $R\subseteq [n]\backslash E$. By Lemma
\ref{lem-ELRC}, $\mathcal C$ is an $(n,k,r,3)$-ELRC.
\end{proof}

\subsection{Optimal $(n,k,r,2)$-ELRC}
In this subsection, we give a method for constructing
$(n=k+\lceil\frac{2k}{r}\rceil,k,r,2)$-ELRC. Our construction is
based on the following lemma.

\begin{lem}\label{subset-1}
Suppose $\left\lfloor\frac{k}{r}\right\rfloor\geq r$. There exists
a collection $\mathcal A=\{A_1,\cdots,A_\eta\}$ of
$\eta=\left\lceil\frac{2k}{r}\right\rceil$ subsets of $[k]$ such
that:
\begin{itemize}
 \item [(\romannumeral 1)] $|A_i|\leq r$ for each $i\in[\eta]$;
 \item [(\romannumeral 2)] $|A_i\cap A_j|\leq 1$ for all $\{i,j\}\subseteq[\eta]$;
 \item [(\romannumeral 3)] Each $i\in[k]$ belongs to exactly two subsets in
 $\mathcal A$;
\end{itemize}
\end{lem}
\begin{proof}
The proof is given in Appendix B.
\end{proof}

The following are two examples of subsets that satisfy conditions
(\romannumeral 1)$-$(\romannumeral 3) of Lemmas \ref{subset-1}.

\begin{exam}\label{exm-t-2-1}
For $k=12$ and $r=3$, we have
$\eta=\left\lceil\frac{2k}{r}\right\rceil=8$. Let $\mathcal
A=\{A_1,\cdots,A_8\}$ be as in Fig. \ref{fg-subset-1}(a), where
each subset in $\{A_1,\cdots,A_4\}$ is represented by a red line
and each subset in $\{A_5,\cdots,A_8\}$ is represented by a blue
line. We can check that conditions (\romannumeral
1)$-$(\romannumeral 3) of Lemmas \ref{subset-1} are satisfied.
\end{exam}

\begin{exam}\label{exm-t-2-2}
For $k=10$ and $r=3$, we have
$\eta=\left\lceil\frac{2k}{r}\right\rceil=7$. Let $\mathcal
A=\{A_1,\cdots,A_7\}$ be as in Fig. \ref{fg-subset-1}(b), where
each subset in $\{A_1,A_2,A_3\}$ is represented by a red solid
line, $A_4$ is represented by a red dashed line and each subset in
$\{A_5,A_6,A_7\}$ is represented by a blue line. We can check that
conditions (\romannumeral 1)$-$(\romannumeral 3) of Lemmas
\ref{subset-1} are satisfied.
\end{exam}

\renewcommand\figurename{Fig}
\begin{figure}[htbp]
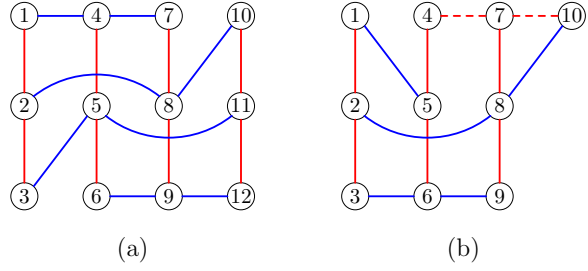

\begin{center}
\includegraphics[height=3.5cm]{eg.1}
\hspace{1cm}\includegraphics[height=3.5cm]{eg.2}
\end{center}
\vspace{-0.2cm}\caption{Subsets of $[k]$ that satisfy conditions
of Lemmas \ref{subset-1}: (a) is for $k=12$ and (b) is for
$k=10$.}\label{fg-subset-1}
\end{figure}

Now we have the following construction.

\textbf{Construction 1}: Let
$\left\lfloor\frac{k}{r}\right\rfloor\geq r$ and $\mathcal
A=\{A_1,\cdots,A_\eta\}$ be constructed as in Lemma
\ref{subset-1}, where $\eta=\left\lceil\frac{2k}{r}\right\rceil$.
Let $x_{1}, \cdots, x_k$ be $k$ information symbols. Then we can
construct a $[k+\eta,k]$ systematic linear code $\mathcal C$ over
$\mathbb F_2$ with $\eta$ parities $x_{k+1}, \cdots, x_{k+\eta}$
such that $x_{k+i}=\sum_{j\in A_i}x_j$ for each $i\in[\eta]$.

\begin{thm}\label{ach-bnd-t-2}
The code $\mathcal C$ obtained by Construction 1 is an
$(n=k+\left\lceil\frac{2k}{r}\right\rceil,k,r,2)$-ELRC.
\end{thm}
\begin{proof}
Let $S=[k]$ and $T=\{k+1,\cdots,k+\eta\}$, where
$\eta=\left\lceil\frac{2k}{r}\right\rceil$. Then we have $S\cap
T=\emptyset$. By conditions (\romannumeral 2), (\romannumeral 3)
of Lemma \ref{subset-1}, for each $i\in S$, there exist two
subsets, say $A_{i_1}$ and $A_{i_2}$, such that $A_{i_1}\cap
A_{i_2}=\{i\}$. By Construction 1 and condition (\romannumeral 1)
of Lemma \ref{subset-1}, $R_1=A_{i_1}\cup\{k+i_1\}\backslash
\{i\}$ and $R_2=A_{i_2}\cup\{k+i_2\}\backslash \{i\}$ are two
disjoint $(r,\mathcal C)$-repair sets of $i$. Moreover, for each
$i\in T$, again by Construction 1 and condition (\romannumeral 1)
of Lemma \ref{subset-1}, $A_{i-k}$ is an $(r,\mathcal C)$-repair
set of $i$. So by Lemma \ref{be-2-elrc}, $\mathcal C$ is an
$(n,k,r,2)$-ELRC.
\end{proof}

Note that the code $\mathcal C$ obtained by Construction 1 has
length $n=k+\eta=k+\left\lceil\frac{2k}{r}\right\rceil$, which
meets the bound \eqref{eq-bnd-t-1-2}. So from Theorem
\ref{ach-bnd-t-2}, we can directly obtain the following theorem.
\begin{thm}\label{ext-bnd-t-2}
If $\left\lfloor\frac{k}{r}\right\rfloor\geq r$, then there exist
$(n,k,r,2)$-ELRC over the binary field that meet the bound
\eqref{eq-bnd-t-1-2}.
\end{thm}

The authors in \cite{Wang15} constructed binary codes with
all-symbol locality $r$, availability $t$ and code rate
$\frac{r}{r+t}$ for $n={r+t\choose r}$ and any positive integer
$r$ and $t~($such codes are a subclass of $(n,k,r,t)$-ELRC$)$. For
$t=2$, we have $n=\frac{(r+2)(r+1)}{2}$ and
$k=\frac{r}{r+2}n=\frac{r(r+1)}{2}$. In our construction, we
require that $\lfloor\frac{k}{r}\rfloor\geq r$, which implies that
$k\geq r^2>\frac{r(r+1)}{2}$ if $r>1$.


\subsection{Optimal $(n,k,r,3)$-ELRC}
In this subsection, we give a method for constructing
$(n=k+\left\lceil\frac{2k+\lceil\frac{k}{r}\rceil}{r}\right\rceil,
k, r, 3)$-ELRC. We always denote
$$m=\left\lceil\frac{k}{r}\right\rceil$$ and
$$\ell=\left\lceil\frac{2k+\lceil\frac{k}{r}\rceil}{r}\right\rceil
-\left\lceil\frac{k}{r}\right\rceil=\left\lceil\frac{2k+m}{r}\right\rceil-m.$$
Then we have
$$n=k+\left\lceil\frac{2k+\lceil\frac{k}{r}\rceil}{r}\right\rceil=k+m+\ell.$$
Our construction is closely related to the following concept.
\begin{defn}\label{def-mesh}
A mesh of $[n]$ is a collection $\mathcal R\cup\mathcal B$ of
subsets of $[n]$, where $\mathcal R=\{RL_1, \cdots, RL_{m}\}$ and
$\mathcal B=\{BL_1,\cdots, BL_{\ell}\}$ are called red lines and
blue lines respectively, that satisfies the following conditions:
\begin{itemize}
 \item [(\romannumeral 1)] For each $i\in[m]$, $RL_i\subseteq[k+m]$, $|RL_i|=r+1$
 and $RL_i\cap\{k+1,\cdots,k+m\}=\{k+i\}$;
 \item [(\romannumeral 2)] For each $j\in[\ell]$,
 $BL_j\cap\{k+m+1, \cdots, n\}=\{k+m+j\}$ and $|BL_j|\leq r+1$;
 \item [(\romannumeral 3)] Each $i\in[k+m]$ belongs to exactly two lines,
 at least one is a red line;
 \item [(\romannumeral 4)] Any two different lines have at most one point in
 common;
 \item [(\romannumeral 5)] Any two different lines that intersect with the same
 red line are disjoint.
\end{itemize}
Here a line means a subset in $\mathcal R\cup\mathcal B$ $($i.e.,
a red line or a blue line$)$ and a point means an element of
$[n]$.
\end{defn}

\renewcommand\figurename{Fig}
\begin{figure}[htbp]
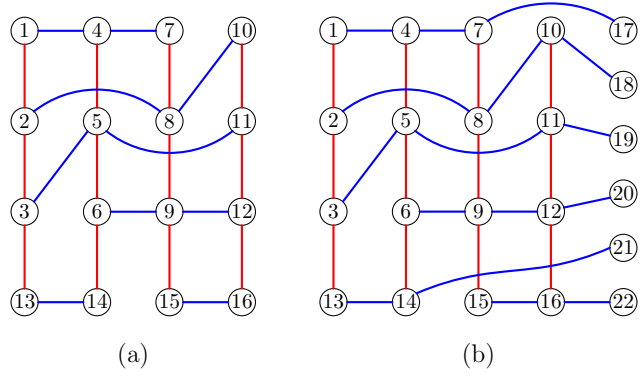

\begin{center}
\includegraphics[height=4.7cm]{eg.3}
\hspace{0.7cm}\includegraphics[height=4.9cm]{eg.4}
\end{center}
\vspace{-0.2cm}\caption{Construction of a mesh of $[n]$, where
$k=12, r=3$ and $n=22$.}\label{fg-subset-2}
\end{figure}

\renewcommand\figurename{Fig}
\begin{figure*}[htbp]
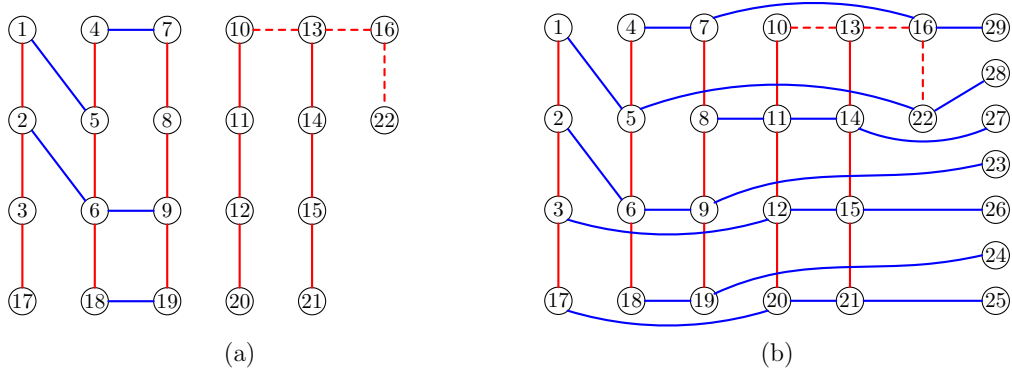

\begin{center}
\includegraphics[height=4.7cm]{eg.5}\vspace{0.5cm}
\hspace{0.8cm}\includegraphics[height=4.9cm]{eg.6}
\end{center}
\vspace{-0.2cm}\caption{Construction of a mesh of $[n]$, where
$k=16, r=3$ and $n=29$.}\label{fg-subset-3}
\end{figure*}

\begin{exam}\label{eg-line-1}
For $k=12$ and $r=3$, we have $m=\lceil\frac{k}{r}\rceil=4$,
$\ell=\left\lceil\frac{2k+\lceil\frac{k}{r}\rceil}{r}\right\rceil
-\left\lceil\frac{k}{r}\right\rceil=6$ and $n=k+m+\ell=22$. Let
$\mathcal R=\{RL_1,\cdots,RL_4\}$ be the red lines and
$\{B_1,\cdots,B_6\}$ be the blue lines in Fig.
\ref{fg-subset-2}(a). Then extend each $B_i$ to a blue line $BL_i$
as in Fig. \ref{fg-subset-2}(b). Let $\mathcal
B=\{BL_1,\cdots,BL_6\}$. We can check that $\mathcal R\cup\mathcal
B$ is a mesh of $[n]$.
\end{exam}

\begin{exam}\label{eg-line-2}
For $k=16$ and $r=3$, we have $m=6$, $\ell=7$ and $n=29$. Let
$\mathcal R=\{RL_1,\cdots,RL_5,RL_6\}$, where $RL_1,\cdots,RL_5$
are the red solid lines in Fig. \ref{fg-subset-3}(a) and $R_6$ is
the red dashed line in Fig. \ref{fg-subset-3}(a). We partition the
first three columns into $B_1=\{2,6,9\}$, $B_2=\{18,19\}$,
$B_3=\{17\}$, $B_4=\{3\}$, $B_5=\{8\}$, $B_6=\{1,5\}$ and
$B_7=\{4,7\}$. In Fig. \ref{fg-subset-3}(a), each $B_i$ of size
$|B_i|\geq 2$ is represented by a blue line and the other points
of the first three columns represent the $B_i$s of size $1$.
Further, we extend each $B_i$ to a blue line $BL_i$ as in Fig.
\ref{fg-subset-3}(b). Let $\mathcal B=\{BL_1,\cdots,BL_6\}$. Then
we can check that $\mathcal R\cup\mathcal B$ is a mesh of $[n]$.
\end{exam}

The following two lemmas and their proofs give some constructions
of mesh of $[n]$.
\begin{lem}\label{line-1}
If $r|k$ and $m\geq r$, there exists a mesh of $[n]$.
\end{lem}
\begin{proof}
The proof is given in Appendix C.
\end{proof}

\begin{lem}\label{line-2}
Suppose $\lambda=r\mod k >0$. If $\ell\geq r+\lambda+1$ and $m\geq
2r-\lambda+1$, then there exists a mesh of $[n]$.
\end{lem}
\begin{proof}
The proof is given in Appendix D.
\end{proof}

Now, we have the following construction.

\textbf{Construction 2}: Let $\mathcal R\cup\mathcal B$ be a mesh
of $[n]$, where $\mathcal R=\{RL_1, \cdots, RL_{m}\}$ is the set
of red lines and $\mathcal B=\{BL_1,\cdots, BL_{\ell}\}$ is the
set of blue lines. Let $x_1,\cdots,x_k$ be $k$ information
symbols. Then we can construct an $[n=k+m+\ell,k]$ systematic
linear code $\mathcal C$ over $\mathbb F_2$ such that the parities
are $x_{k+1},\cdots,x_n$ and are computed as follows:
\begin{itemize}
 \item For each $i\in[m]$,
 \begin{align}\label{cst-2-1}x_{k+i}=\sum_{j\in
 RL_i\backslash\{k+i\}}x_{j}.\end{align}
 \item For each $i\in[\ell]$, \begin{align}\label{cst-2-2}x_{k+m+i}
 =\sum_{j\in BL_i\backslash\{k+m+i\}}x_{j}.\end{align}
\end{itemize}

Note that by condition (\romannumeral 1) of Definition
\ref{def-mesh}, for each $i\in[m]$, we have $RL_{i}\backslash
\{k+i\}\subseteq[k]$. So by \eqref{cst-2-1}, $x_{k+i}$ is
computable from information symbols. Similarly, for each
$i\in[\ell]$, by condition (\romannumeral 2) of Definition
\ref{def-mesh}, $BL_{i}\backslash\{k+m+i\}\subseteq[k+m]$. So by
\eqref{cst-2-2}, $x_{k+m+i}$ is computable from
$\{x_j;j\in[k+m]\}$. Hence, Construction 2 is reasonable.

\begin{thm}\label{ach-bnd-t-3}
The code $\mathcal C$ obtained by Construction 2 is an
$(n=k+m+\ell,k,r,3)$-ELRC.
\end{thm}
\begin{proof}
Let $S=[k+m]$ and $T=\{k+m+1,\cdots,n\}$. Then $S\cap
T=\emptyset$.

For each $i\in S$, by conditions (\romannumeral 3) and
(\romannumeral 4) of Definition \ref{def-mesh}, there exists a red
line $L\in\mathcal R$ and a line $L'\in\mathcal R\cup\mathcal B$
such that $L\cap L'=\{i\}$. By conditions (\romannumeral 1),
(\romannumeral 2) of Definition \ref{def-mesh},
$|L\backslash\{i\}|=r$ and $|L'\backslash\{i\}|\leq r$. So by
\eqref{cst-2-1} and \eqref{cst-2-2}, $R_1=L\backslash\{i\}$ and
$R_2=L'\backslash\{i\}$ are two disjoint $(r,\mathcal C)$-repair
sets of $i$. Moreover, for each $j\in L\backslash\{i\}$, by
condition (\romannumeral 1) of Definition \ref{def-mesh}, $j\in
L\subseteq[k+m]$. Then by condition (\romannumeral 3) of
Definition \ref{def-mesh}, there exists an $L''\in\mathcal
R\cup\mathcal B$ such that $L''\neq L$ and $j\in L''$. Clearly,
$L''\neq L'$. $($Otherwise, $\{i,j\}\subseteq L\cap L'= L\cap
L''$, which contradicts to condition (\romannumeral 4) of
Definition \ref{def-mesh}.$)$ So by condition (\romannumeral 5) of
Definition \ref{def-mesh}, $L''\cap L'=\emptyset$. Let
$R=L''\backslash\{j\}$. Then $R\cap(R_2\cup\{i\})\subseteq L''\cap
L'=\emptyset$ and by \eqref{cst-2-1}, \eqref{cst-2-2}, $R$ is an
$(r,\mathcal C)$-repair set of $j$.

For each $i\in T$, let $i'=i-(k+m)$. Then $i'\in[\ell]$. Let
$R=BL_{i'}\backslash\{i\}$. Then by condition (\romannumeral 2) of
Definition \ref{def-mesh} and by \eqref{cst-2-2},
$R\subseteq[k+m]=S$ is an $(r,\mathcal C)$-repair sets of $i$.

By Lemma \ref{be-3-elrc}, $\mathcal C$ is an $(n,k,r,3)$-ELRC.
\end{proof}

Note that the code $\mathcal C$ obtained by Construction 2 has
length
$n=k+m+\ell=k+\left\lceil\frac{2k+\lceil\frac{k}{r}\rceil}{r}\right\rceil$,
which meets the bound \eqref{eq-bnd-t-3}. So the following theorem
is a direct consequence of Lemma \ref{line-1}, \ref{line-2} and
Theorem \ref{ach-bnd-t-3}.
\begin{thm}\label{ext-bnd-t-3}
Suppose one of the following conditions hold:
\begin{itemize}
 \item [(\romannumeral 1)] $r|k$ and $m\geq r$.
 \item [(\romannumeral 2)] $\ell \geq r+\lambda+1$ and $m\geq 2r-\lambda+1$,
  where $\lambda = r\mod k >0$.
\end{itemize}
Then there exist $(n,k,r,3)$-ELRC over the binary field that meet
the bound \eqref{eq-bnd-t-3}.
\end{thm}

Binary codes with all-symbol locality $r$, availability $t$ and
code rate $\frac{r}{r+t}$ are constructed in \cite{Wang15} for any
positive integers $r$ and $t$ (such codes are a subclass of
$(n,k,r,t)$-ELRC). For $t=3$, the code length is
$n=k\frac{r+3}{r}=k+\frac{3k}{r}>
k+\left\lceil\frac{2k+\lceil\frac{k}{r}\rceil}{r}\right\rceil$.
Hence is not optimal according to the bound \eqref{eq-bnd-t-3}.

\section{Conclusions}
We investigate the problem of coding for distributed storage
system that can locally repair up to $t$ failed nodes, where $t$
is a given positive integer. Given the code dimension $k$, the
repair locality $r$ and $t\in\{2,3\}$, we derive a lower bound on
the code length $n$ under the functional repair model. We also
give some constructions of exact LRCs for $t\in\{2,3\}$ with
binary field and whose length $n$ achieves the corresponding
bounds, which proves the tightness of our bounds and also implies
that there is no gap between the optimal code length of functional
LRCs and exact LRCs for certain sets of parameters.

Some problems are still open. For example, what is the optimal
code length for $t\geq 4$? Given $n,k,r$ and $t$, what is the
upper bound of the minimum distance $d$? Another interesting
problem is to construct functional locally repairable codes
$\{\mathcal C_\lambda; \lambda\in\Lambda\}$ with small size of
$\Lambda$.

\appendices

\section{Proof of Claim 1}
To prove Claim 1, the key is to prove the following two
statements: a) For each $v\in B\cup C_1$, $|\mathcal
E_{\text{blue}}(v)|\geq 1$; b) Each blue edge belongs to at most
$r$ different $v\in B\cup C_1$.

For each $v\in B$, by \eqref{divide-b}, $|\text{Out}(v)|=2$. So we
can assume $\text{Out}(v)=\{v_1,v_2\}$. Then $v_1, v_2$ are two
inner nodes of $G_{\lambda_0}$. By 1) of Corollary
\ref{mrg-cor-2}, $\text{Out}(v_1)\neq\emptyset$ or
$\text{Out}(v_2)\neq\emptyset$. Without loss of generality, we can
assume $\text{Out}(v_1)\neq\emptyset$ and $v_3\in\text{Out}(v_1)$.
Then we have the following two cases:

Case 1: $(v_1,v_3)$ is not a green edge. Since $v_1$ is an inner
node, then $(v_1,v_3)$ is not a red edge. Note that $v\in B$ and
$v_1\in\text{Out}(v)$. Then $(v_1,v_3)$ is a blue edge belonging
to $v$.

Case 2: $(v_1,v_3)$ is a green edge. Then $\{v_1\}=\text{Out}(u)$
for some $u\in C_1\cup C_2$. By 2) of Corollary \ref{mrg-cor-2},
$\text{Out}(v_2)\neq\emptyset$. Let $v_4\in\text{Out}(v_2)$. Since
$v_2$ is an inner node, then $(v_2,v_4)$ is not a red edge. Not
that by 3) of Corollary \ref{mrg-cor-2}, $|\text{Out}(w)|\geq 2$
for any source $w\in\text{In}(v_2)$. (As illustrated in Fig.
\ref{fg-clm-1}(a).) Then $(C_1\cup
C_2)\cap\text{In}(v_2)=\emptyset$, which implies that
$v_2\notin\text{Out}(m)$ for any $m\in C_1\cup C_2$. So
$(v_2,v_4)$ is not a green edge. Since $v\in B$ and
$v_2\in\text{Out}(v)$, then $(v_2,v_4)$ is a blue edge belonging
to $v$.

\renewcommand\figurename{Fig}
\begin{figure}[htbp]
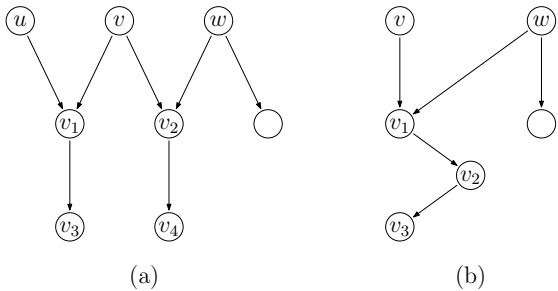

\begin{center}
\includegraphics[height=3.8cm]{fg1.2}
\hspace{1.2cm}\includegraphics[height=3.8cm]{fg1.1}
\end{center}
\vspace{-0.2cm}\caption{Illustration of the local graph in the
proof of Claim 1.}\label{fg-clm-1}
\end{figure}

In both cases, we can find a blue edge belonging to $v$.

For each $v\in C_1$, by \eqref{divide-c1},
$|\text{Out}(v)|=|\text{Out}^2(v)|=1$. We can assume
$\text{Out}(v)=\{v_1\}$ and $\text{Out}^2(v)=\{v_2\}$. Then $v_1,
v_2$ are two inner nodes. By 2) of Corollary \ref{mrg-cor-1}, we
have $\text{Out}^2(v)=\text{Out}(v_1)=\{v_2\}$. Further, by 3) of
Corollary \ref{mrg-cor-1}, we have $\text{Out}(v_2)\neq\emptyset$.
Let $v_3\in\text{Out}(v_2)$. Since $v_2$ is an inner node, the
edge $(v_2,v_3)$ is not a red edge. Not that by 4) of Corollary
\ref{mrg-cor-1}, $|\text{Out}(u)|\geq 2$ for any source
$u\in\text{In}(v_2)$. (As illustrated in Fig. \ref{fg-clm-1}(b).)
Then we have $(C_1\cup C_2)\cap\text{In}(v_2)=\emptyset$, which
implies that $v_2\notin\text{Out}(u)$ for any $u\in C_1\cup C_2$.
So $(v_2,v_3)$ is not a green edge. Note that $v\in C_1$ and
$\text{Out}^2(v)=\text{Out}(v_1)=\{v_2\}$. So $(v_2,v_3)$ is a
blue edge belonging to $v$.

By the above discussion, we proved that $|\mathcal
E_{\text{blue}}(v)|\geq 1$ for each $v\in B\cup C_1$, which proves
the statement a).

Let $(u',u'')$ be a blue edge and $S$ be the set of all $v\in
B\cup C_1$ such that $(u',u'')$ belongs to $v$. 
For each $v\in S$, we pick a $\varphi(v)\in\text{In}(u')$
depending on the following two cases:

Case 1: $v\in B$. Since $(u',u'')$ is a blue edge belongs to $v$,
then $u'\in\text{Out}(v)$, which implies $v\in\text{In}(u')$. Pick
$\varphi(v)=v$.

Case 2: $v\in C_1$. By \eqref{divide-c1},
$|\text{Out}^2(v)|=|\text{Out}(v)|=1$. Denote
$\text{Out}(v)=\{v'\}$. Then by 2) of Corollary \ref{mrg-cor-1},
$\text{Out}^2(v)=\text{Out}(v')$. Moreover, since $(u',u'')$ is a
blue edge belongs to $v$, then
$u'\in\text{Out}^2(v)=\text{Out}(v')$. So $v'\in\text{In}(u')$.
Pick $\varphi(v)=v'$.

If $v$ and $w$ are two different sources in $S\cap C_1$, by 5) of
Corollary \ref{mrg-cor-1}, their out-neighbors are different. So
$\varphi(v)\neq\varphi(w)$. Thus, $\varphi$ is a one-to-one
correspondence between $S$ and a subset of $\text{In}(u')$. Note
that $|\text{In}(u')|\leq r$. So $|S|\leq|\text{In}(u')|\leq r$.
Thus, $(u',u'')$ belongs to at most $r$ different $v\in B\cup
C_1$, which proves the statement b).

By statements a) and b), we have $|\mathcal
E_{\text{blue}}|\geq\frac{|B|+|C_1|}{r}$, which proves Claim 1.

\section{Proof of Lemma \ref{subset-1}}
We need to consider two cases, i.e., $r|k$ and $r\nmid k$.

Case 1: $r|k$. We can let $k=mr$. Then
$\eta=\left\lceil\frac{2k}{r}\right\rceil=2m$ and
$m=\left\lfloor\frac{k}{r}\right\rfloor$. By assumption of Lemma
\ref{subset-1}, $m=\left\lfloor\frac{k}{r}\right\rfloor\geq r$. We
assign the elements of $[k]$ in a $r\times m$ array
$D=(a_{i,j})_{i\in[r],j\in[m]}$ as in Fig. \ref{fg-subset-1-1}
such that $[k]=\{a_{i,j}; i\in[r],j\in[m]\}.$ For each $j\in[m]$,
let $A_j=\{a_{i,j};i\in[r]\}$. Then $|A_i|=r, \forall i\in[m]$. In
Fig. \ref{fg-subset-1-1}, each subset $A_i$ is represented by a
red line.

\renewcommand\figurename{Fig}
\begin{figure}[htbp]
\begin{center}
\includegraphics[height=2.5cm]{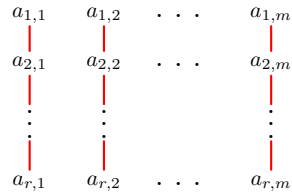}
\end{center}
\vspace{-0.2cm}\caption{Partition of $[n]$: Each subset is
represented by a red line.}\label{fg-subset-1-1}
\end{figure}

Let $\delta=r$ and $\mathcal L=\{A_1,\cdots,A_m\}$. Then
$|A_i|=\delta$ for each $i\in[m]$. Let $r_i=r, \forall i\in[m]$.
Then $\sum_{i=1}^mr_i=mr=\sum_{j=1}^m|A_j|$. Since $m\geq r=r_i,
\forall i\in[m]$, then by Corollary \ref{subset-eq-c}, there
exists a collection $\{B_1,\cdots,B_{m}\}$ of subsets of
$\bigcup_{j=1}^{m}A_j=[k]$ that satisfies the following three
properties:
\begin{itemize}
 \item $B_1,\cdots,B_{m}$ are pairwise disjoint and
 $\bigcup_{i=1}^{m}B_i=\bigcup_{j=1}^{m}A_j=[k]$;
 \vspace{0.05cm}\item $|B_i|=r_i=r$ for all $i\in[m]$;
 \vspace{0.05cm}\item $|B_i\cap A_j|\leq 1$ for all $i,j\in[m]$.
\end{itemize}
For each $i\in[m]$, let $A_{m+i}=B_{i}$. Then it is easy to check
that $\mathcal A=\{A_1,\cdots,A_{\eta}\}$ satisfies conditions
(\romannumeral 1)$-$(\romannumeral 3) of Lemma \ref{subset-1},
where $\eta=\left\lceil\frac{2k}{r}\right\rceil=2m$.

Case 2: $r\nmid k$. Let $m=\lceil\frac{k}{r}\rceil$. Since $r\nmid
k$, then $m-1=\left\lfloor\frac{k}{r}\right\rfloor$ and
$k=(m-1)r+\lambda$, where $0<\lambda<r$. By assumption of Lemma
\ref{subset-1}, we have
$$m-1=\left\lfloor\frac{k}{r}\right\rfloor\geq r.$$
Let $\alpha=m-1-(r-\lambda)$. We can assign elements of $[k]$ in
an $r\times m$ array $D=(a_{i,j})_{i\in[r],j\in[m]}$ as in Fig.
\ref{fg-subset-1-2} such that $\{a_{i,j};
i\in[r],j\in[m-1]\}\cup\{a_{i,m}; i\in[\lambda]\}=[k]$ and
$a_{i,m}=0, ~ \forall i\in\{\lambda+1,\cdots,r\}.$ Let
$$A_0=\{a_{1,j}; j\in\{\alpha+1,\cdots,m-1\}\}.$$ Then
$|A_0|=(m-1)-\alpha=r-\lambda$. Let
\begin{equation*}
A_j=\left\{\begin{aligned}
&\{a_{i,j};i\in[r]\}, {\large ~ ~ ~ ~} ~ ~ ~ ~ \text{if}~j\in[m-1];\\
&\{a_{i,m}; i\in[\lambda]\}\cup A_0, ~ \text{if}~j=m.\\
\end{aligned} \right. \label{eqn:2}
\end{equation*} In
Fig. \ref{fg-subset-1-2}, each subset in $\{A_1,\cdots,A_{m-1}\}$
is represented by a red solid line and $A_m$ is represented by a
red dashed line. For convenience, we call each subset in
$\{A_1,\cdots,A_m\}$ a red line. Clearly, $|A_{j}|=r$ and
$|A_j\cap A_{j'}|\leq 1$ for all $j\neq j'\in[m]$.

\renewcommand\figurename{Fig}
\begin{figure}[htbp]
\begin{center}
\includegraphics[height=4.2cm]{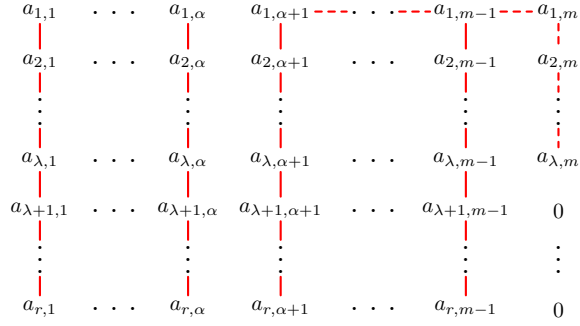}
\end{center}
\vspace{-0.2cm}\caption{Construction of subsets: Each of the first
$m-1$ subsets is represented by a red solid line and the $m$th
subset is represented by a red dashed line.}\label{fg-subset-1-2}
\end{figure}

For each $j\in[m]$, let $C_j=A_j\backslash A_0$. Then
$C_1,\cdots,C_{m}$ are pairwise disjoint and
$\bigcup_{j=1}^{m}C_j=[k]\backslash A_0$. So
$|\bigcup_{j=1}^{m}C_j|=|[k]\backslash A_0|=k-r+\lambda$.
Moreover, we have
\begin{equation*}
|C_j|=\left\{\begin{aligned}
&r, ~ ~ ~ ~ ~ ~ ~ \text{if}~j\in[\alpha];\\
&r-1, ~ ~ \text{if}~j\in\{\alpha+1,\cdots,m-1\};\\
&\lambda, ~ ~ ~ ~ ~ ~ ~ \text{if}~j=m.\\
\end{aligned} \right. \label{eqn:2}
\end{equation*}
Let $\rho=\left\lceil\frac{k-r+\lambda}{r}\right\rceil$. Then
$k-r+\lambda$ can be represented as the sum of $\rho$ positive
integers (not necessarily different) $r_1,\cdots,r_\rho$ such that
$r_i\leq r, \forall i\in[\rho]$. Since $m-1\geq r$, using the
Gale-Ryser Theorem, we can construct an $m\times\rho$ binary
matrix $M$ such that for each $i\in[\rho]$ and each $j\in[m]$, the
sum of the $i$th row is $r_i$ and the sum of the $j$th column is
$|C_j|$. Let $\mathcal L=\{C_1,\cdots,C_{m}\}$. By Lemma
\ref{M-2-S}, there exists a collection $\{B_1,\cdots,B_{\rho}\}$
of subsets of $\bigcup_{j=1}^{m}C_j=[k]\backslash A_0$ such that
\begin{itemize}
 \item $B_1,\cdots,B_{\rho}$ are pairwise disjoint and
 $\bigcup_{i=1}^{\rho}B_i=\bigcup_{j=1}^{m}C_j=[k]\backslash A_0$;
 \vspace{0.2cm}\item $|B_i|=r_i$ for all $i\in[\rho]$;
 \vspace{0.1cm}\item $|B_i\cap C_j|\leq 1$ for all $i\in[\rho]$ and $j\in[m]$.
\end{itemize}
Now, for each $i\in[\rho]$, let $A_{m+i}=B_i$. Note that
$k=(m-1)r+\lambda$ and
$\rho=\left\lceil\frac{k-r+\lambda}{r}\right\rceil$. Then
$m+\rho=m+\left\lceil\frac{k-r+\lambda}{r}\right\rceil=
\left\lceil\frac{mr+k-r+\lambda}{r}\right\rceil
=\left\lceil\frac{2k}{r}\right\rceil=\eta$. Thus, we obtain a
collection $\mathcal A=\{A_1,\cdots,A_{\eta}\}$ of $\eta$ subsets
of $[k]$. For convenience, we call each subset in
$\{A_{m+1},\cdots,A_\eta\}$ a blue line.

By the construction, we have $|A_i|\leq r$ for each $i\in[\eta]$.
So condition (\romannumeral 1) of Lemma \ref{subset-1} is
satisfied.

Again by the construction, we have the following observations: 1)
Each $i\in A_0$ belongs to exactly two red lines and each
$i\in[k]\backslash A_0$ belongs to one red line and one blue line;
2) Any two different red lines has at most one point (element) in
common; 3) Any two different blue lines have no point (element) in
common; 4) A red line and a blue line have at most one point
(element) in common.

Observation 1) implies that each $i\in[k]$ belongs to exactly two
subsets in $\mathcal A$. So condition (\romannumeral 3) of Lemma
\ref{subset-1} is satisfied. Moreover, observations 2)$-$4) imply
that any two different lines have at most one point (element) in
common. So condition (\romannumeral 2) of Lemma \ref{subset-1} is
satisfied.

Thus, we can always construct a collection of
$\eta=\left\lceil\frac{2k}{r}\right\rceil$ subsets of $[k]$ that
satisfies conditions (\romannumeral 1)$-$(\romannumeral 3) of
Lemma \ref{subset-1}.

\section{Proof of Lemma \ref{line-1}}
We will construct a set $\mathcal R=\{RL_1, \cdots, RL_{m}\}$ of
red lines and a set $\mathcal B=\{BL_1,\cdots, BL_{\ell}\}$ of
blue lines and prove that $\mathcal R\cup\mathcal B$ is a mesh of
$[n]$.

Since $m=\left\lceil\frac{k}{r}\right\rceil$ and by assumption of
Lemma \ref{line-1}, $r|k$, then $k=mr$ and $k+m=(r+1)m$. We can
assign the elements of $[k+m]$ in an $(r+1)\times m$ array
$D=(a_{i,j})_{i\in[r+1],j\in[m]}$ as in Fig. \ref{fg-line-1} such
that $[k]=\{a_{i,j};i\in[r],j\in[m]\}$ and $a_{r+1,j}=k+j,
~\forall j\in[m]$. For each $j\in[m]$, we let $RL_j=\{a_{i,j};
i\in[r+1]\}$. In Fig. \ref{fg-line-1}, each subset in
$\{RL_1,\cdots,RL_{m}\}$ is represented by a red solid line.

\renewcommand\figurename{Fig}
\begin{figure}[htbp]
\begin{center}
\includegraphics[height=3.0cm]{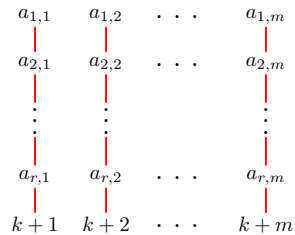}
\end{center}
\vspace{-0.2cm}\caption{Construction of red lines: Each red line
is a column of the array.}\label{fg-line-1}
\end{figure}

Since $k=mr$, then $\ell=\left\lceil\frac{2k+m}{r}\right\rceil-m=
\left\lceil\frac{k+m}{r}\right\rceil$. Hence, $k+m$ can be
represented as the sum of $\ell$ positive integers
$r_1,\cdots,r_{\ell}$ such that $r_i\leq r$ for each $i\in[\ell]$.
Let $\mathcal L=\{RL_1,\cdots,RL_m\}$ and $\delta=r+1$. Note that
by assumption of Lemma \ref{line-1}, $m\geq r$. So we have
$r_i\leq r\leq m$ for each $i\in[\ell]$. By Corollary
\ref{subset-eq-c}, there exists a collection
$\{B_1,\cdots,B_{\ell}\}$ of subsets of $\bigcup_{j=1}^{m}RL_j$
that satisfies the following properties:
\begin{itemize}
 \item $B_1,\cdots,B_{\ell}$ are pairwise disjoint and
 $\bigcup_{i=1}^{\ell}B_i=\bigcup_{j=1}^{m}RL_j=[k+m]$;
 \vspace{0.05cm}\item $|B_i|=r_i$ for all $i\in[\ell]$;
 \vspace{0.05cm}\item $|B_i\cap RL_j|\leq 1$ for all
 $i\in[\ell]$ and $j\in[m]$.
\end{itemize}
For each $i\in[\ell]$, let $BL_{i}=B_{i}\cup\{k+m+i\}$, and let
$\mathcal B=\{BL_1,\cdots, BL_{\ell}\}$.

By the construction, it is easy to check that conditions
(\romannumeral 1), (\romannumeral 2), (\romannumeral 4) of
Definition \ref{def-mesh} are satisfied.

By the construction, we also have the following observations: 1)
$\mathcal R$ is a partition of $[k+m]$; 2) $\mathcal B$ is a
partition of $[n]$; 3) $|BL_i\cap RL_j|\leq 1$ for all
$i\in[\ell]$ and $j\in[m]$.

By the above observations, we can easily check that conditions
(\romannumeral 3), (\romannumeral 5) of Definition \ref{def-mesh}
are satisfied.

So $\mathcal R\cup\mathcal B$ is a mesh of $[n]$.

\section{Proof of Lemma \ref{line-2}}
We will construct a set $\mathcal R=\{RL_1, \cdots, RL_{m}\}$ of
red lines and a set $\mathcal B=\{BL_1,\cdots, BL_{\ell}\}$ of
blue lines and prove that $\mathcal R\cup\mathcal B$ is a mesh of
$[n]$.

Since $m=\left\lceil\frac{k}{r}\right\rceil$ and $\lambda=r\mod k
>0$, then
\begin{align}\label{line-2-eq-k}k=(m-1)r+\lambda.\end{align}
Hence, $k+m=(m-1)r+\lambda+m=(m-1)(r+1)+(\lambda+1)$. We can
assign the elements of $[k+m]$ in an $(r+1)\times m$ array
$D=(a_{i,j})_{i\in[r+1],j\in[m+1]}$ as in Fig. \ref{fg-line-2}
such that
$[k+m]=\{a_{i,j};i\in[r+1],j\in[m-1]\}\cup\{a_{i,m};i\in[\lambda+1]\}$
and $a_{i,m+1}=0$ for $i\in\{\lambda+2,\cdots,r+1\}$. Moreover, by
proper permutation (if necessary), we can let $a_{r+1,j}=k+j$ for
each $j\in[m-1]$ and $a_{\lambda+1,m}=k+m.$ We can construct
$\mathcal R=\{RL_1, \cdots, RL_{m}\}$ and $\mathcal
B=\{BL_1,\cdots, BL_{m+\lambda}\}$ by the following three steps.

\renewcommand\figurename{Fig}
\begin{figure}[htbp]
\begin{center}
\includegraphics[height=5.3cm]{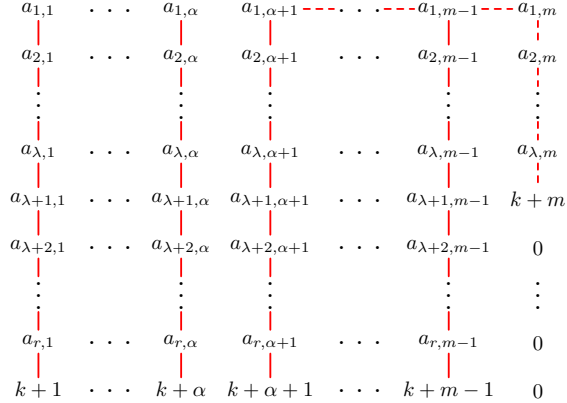}
\end{center}
\vspace{-0.2cm}\caption{Construction of red lines of $[n]$: The
first $m-1$ red lines are the first $m-1$ columns of the array and
the last red line is depicted by a dashed red line, where
$\alpha=m-1-(r-\lambda)$.}\label{fg-line-2}
\end{figure}

\textbf{Step 1}: Construct $\mathcal R=\{RL_1,\cdots,RL_m\}$.

Denote
\begin{align}\label{alf-def}
\alpha=m-1-(r-\lambda)
\end{align}
and for each $i\in[r+1]$, let
$$A_i=\{a_{i,j}; j\in\{\alpha+1,\cdots,m-1\}\}.$$ Then we have
$|A_i|=m-1-\alpha=r-\lambda, \forall i\in[r+1]$.

For each $j\in[m]$, let
\begin{equation*} RL_j=\left\{\begin{aligned}
&\{a_{i,j};i\in[r+1]\}, {\large ~ ~ ~ ~} ~ ~ ~ ~ ~ \text{if}~j\in[m-1];\\
&\{a_{i,m};i\in[\lambda+1]\}\cup A_1, ~ ~ \text{if}~j=m+1.\\
\end{aligned} \right. \label{eqn:2}
\end{equation*}
In Fig. \ref{fg-line-2}, each subset in $\{RL_1,\cdots,RL_{m-1}\}$
is represented by a red solid line and $RL_m$ is represented by a
red dashed line. Clearly, $|RL_i|=r+1$ for all $i\in[m-1]$.
Moreover, by the construction,
$|RL_m|=|A_1|+\lambda+1=(r-\lambda)+(\lambda+1)=r+1$. So we have
$|RL_i|=r+1$ for all $i\in[m]$.

\textbf{Step 2}: Partition $\bigcup_{i=1}^\alpha RL_i$.

By assumption of this lemma, $m\geq 2r-\lambda+1$, which implies
that $m-1-(r-\lambda)\geq r$. So by \eqref{alf-def}, we have
$$\alpha=m-1-(r-\lambda)\geq r.$$ Let
\begin{align}\label{bet-def}
\beta=\alpha(r+1)-(\lambda+1)(r-1)-r\lambda
\end{align}
and
\begin{align}\label{h-def}
h=\ell-(\lambda+1)-r.\end{align} By assumption of this lemma,
$\ell\geq\lambda+1+r$. So we have $h\geq 0$. Moreover, note that
\begin{align*}
\left\lceil\frac{\beta}{r}\right\rceil
&=\left\lceil\frac{\alpha(r+1)-r\lambda-(\lambda+1)(r-1)}{r}\right\rceil\\
&=\left\lceil\frac{(m-1-r+\lambda)(r+1)-r\lambda-(\lambda+1)(r-1)}{r}\right\rceil\\
&=\left\lceil\frac{2[(m-1)r+\lambda]+m}{r}-m-(\lambda+1)-r\right\rceil\\
&=\left\lceil\frac{2k+m}{r}\right\rceil-m-(\lambda+1)-r\\
&=\ell-(\lambda+1)-r\\&=h.\end{align*} So $\beta$ can be
represented as the sum of $h$ positive integers, say
$r_1,\cdots,r_h$, such that $r_i\leq r, \forall i\in[h]$.
Moreover, we let
\begin{equation*}
r_i=\left\{\begin{aligned}
&r-1, ~ ~ \text{if}~i\in\{h+1,\cdots,h+\lambda+1\};\\
&\lambda, ~ ~ ~ ~ ~ ~ ~ \text{if}~i\in\{h+\lambda+2,\cdots,\ell\}.\\
\end{aligned} \right. \label{eqn:2}
\end{equation*}
Then by \eqref{bet-def} and \eqref{h-def}, we have
\begin{align*}\sum_{i=1}^{\ell}r_i&=\sum_{i=1}^{h}r_i
+\sum_{i=h+1}^{h+\lambda+1}r_i+\sum_{i=h+\lambda+2}^{\ell}r_i\\
&=\beta+(\lambda+1)(r-1)+(\ell-h-\lambda-1)\lambda\\
&=\beta+(\lambda+1)(r-1)+r\lambda\\&=\alpha(r+1)\\&=\left|\bigcup_{i=1}^\alpha
RL_i\right|.\end{align*}

Let $\mathcal L=\{RL_{1},\cdots,RL_{\alpha}\}$ and $\delta=r+1$.
Note that $r_i\leq r\leq\alpha=|\mathcal L|, \forall i\in[\ell]$.
Then by Corollary \ref{subset-eq-c}, there exists a collection
$\{B_1,\cdots,B_{\ell}\}$ of subsets of
$\bigcup_{i=1}^{\alpha}RL_i$ that satisfies the following three
properties:
\begin{itemize}
 \item $B_1,\cdots,B_{\ell}$ are pairwise disjoint
 and $\bigcup_{i=1}^{\ell}B_i=\bigcup_{i=1}^{\alpha}RL_i$;
 \item $|B_i|=r_i$ for all $i\in[\ell]$;
 \item $|B_i\cap RL_j|\leq 1$ for all $i\in[\ell]$ and
 $j\in[\alpha]$.
\end{itemize}

\textbf{Step 3}: For each $i\in[\ell]$, extend $B_i$ to $BL_i$.


For each $i\in[h]$, let $$BL_i=B_i\cup\{k+m+i\};$$ For each
$i\in\{h+1,\cdots,h+\lambda+1\}$, let
$$BL_i=B_i\cup\{a_{i-h,m+1}, k+m+i\};$$ For each
$i\in\{h+\lambda+2,\cdots,\ell\}$, let
$$BL_i=B_i\cup A_{i-h-\lambda}\cup\{k+m+i\}.$$ Note that by
\eqref{h-def}, we have $\ell-h-\lambda=r+1$. So for each
$i\in\{h+\lambda+2,\cdots,\ell\}$, we have
$i-h-\lambda\in\{2,\cdots,r+1\}$. Hence, $BL_i$ is reasonably
constructed and $A_1\cap BL_i=\emptyset$.

By the construction, it is easy to see that conditions
(\romannumeral 1), (\romannumeral 2) of Definition \ref{def-mesh}
are satisfied. Moreover, we can see that each point in $A_1$
belongs to two red lines and each point in $[k+m]\backslash A_1$
belongs to a red line and a blue line. So condition (\romannumeral
3) of Definition \ref{def-mesh} is satisfied.

By the construction, we also have the following observations: 1)
$|RL_m\cap RL_i|=0$ for $i\in[\alpha]$; 2) $|RL_m\cap RL_j|=1$ for
$j\in\{\alpha+1,\cdots,m-1\}$; 3) If $i,j\in[m-1]$ and $i\neq j$,
then $RL_i$ and $RL_j$ have no point in common; 4) A red line and
a blue line have at most one point in common; 5) Two different
blue lines have no point in common; 6) If a blue line intersects
with $RL_m$, then it does not intersect with $RL_i$ for all
$i\in\{\alpha+1,\cdots,m-1\}$.

Note that observations 1)$-$3) imply that any two different red
lines have at most one point in common. Hence observations 1)$-$5)
imply that condition (\romannumeral 4) of Definition
\ref{def-mesh} is satisfied. Now suppose that two lines, say $L_1$
and $L_2$, intersect with $RL_i$ for some $i\in[m]$. We have the
following three cases:

Case 1: $i\in[\alpha]$. Then by observations 1) and 3), $L_1$ and
$L_2$ are two different blue lines. So by observation 5), $L_1$
and $L_2$ have no point in common.

Case 2: $i\in\{\alpha+1,\cdots,m-1\}$. Then by observations 2) and
3), we have the following two subcases.

Case 2.1: $L_1$ is $RL_m$ and $L_2$ is a blue line. By observation
6), $L_1$ and $L_2$ have no point in common.

Case 2.2: $L_1$ and $L_2$ are two different blue lines. Then by
observation 5), $L_1$ and $L_2$ have no point in common.

Case 3: $i=m$. Then by observations 1) and 2), we have the
following three subcases.

Case 3.1: $L_1$ is $RL_i$ for some $i\in\{\alpha+1,\cdots,m-1\}$
and $L_2$ is a blue line. By observation 6), $L_1$ and $L_2$ have
no point in common.

Case 3.2: $L_1=RL_i$ and $L_2=RL_j$ for some $i,j\in[m-1]$ and
$i\neq j$. By observation 3), $L_1$ and $L_2$ have no point in
common.

Case 3.3: $L_1$ and $L_2$ are two different blue lines. Then by
observation 5), $L_1$ and $L_2$ have no point in common.

By above discussion, we proved that condition (\romannumeral 5) of
Definition \ref{def-mesh} is satisfied.

So $\mathcal R\cup\mathcal B$ is a mesh of $[n]$.

\end{document}